\documentclass[11pt,letterpaper]{article}

\usepackage{url}
\usepackage{graphicx,cite}
\usepackage{color}

\usepackage{amsmath,amssymb}
\usepackage{amsthm}
\usepackage[ruled,vlined,linesnumbered]{algorithm2e}
\usepackage{pgf,tikz}
\usetikzlibrary{arrows}
\usepackage{xspace,units}
\usepackage{typearea}
\usepackage{fullpage}

\allowdisplaybreaks

\usepackage{lmodern}
\usepackage[T1]{fontenc}
\usepackage{textcomp}

\typearea{16}

\newtheorem{lemma}{Lemma}[section]

\newtheorem{theorem}{Theorem}[section]
\newtheorem{corollary}{Corollary}[section]

\newtheorem{claim}{Claim}[section]
\newtheorem*{claim*}{Claim}
\newtheorem{assumption}{Assumption}[section]
\newtheorem{definition}{Definition}[section]
\newtheorem{remark}{Remark}[section]

\newcommand{\vecp}{{\mathbf p}}

\newcommand{\vecx}{{\mathbf x}}
\newcommand{\vecs}{{\mathbf s}}
\newcommand{\vecw}{{\mathbf w}}
\newcommand{\vecy}{{\mathbf y}}
\newcommand{\vecz}{{\mathbf z}}
\newcommand{\R}{{\mathbb R}}
\newcommand{\ts}{{\tilde{s}}}
\newcommand{\tvecs}{{\tilde{\vecs}}}

\newcommand{\ot}{\leftarrow}

\newcommand{\Oracle}{{\mathcal O}\xspace}
\newcommand{\AOracle}{\widetilde{\mathcal O}\xspace}

\newcommand{\classP}{{\sf P}}
\newcommand{\classPPAD}{{\sf PPAD}}

\newcommand{\UpdatePrice}{{\textsc{Update}}\xspace}

\SetAlgoCaptionSeparator{.}

\let\oldnl\nl
\newcommand{\nonl}{\renewcommand{\nl}{\let\nl\oldnl}}

\begin{document}

\title{Ascending-Price Algorithms for Unknown Markets}

\author{Xiaohui Bei
\footnote{Nanyang Technological University, Singapore.}
\and Jugal Garg
\footnote{University of Illinois at Urbana-Champaign, USA.}
\and Martin Hoefer
\footnote{Max-Planck-Institut f\"ur Informatik and Saarland University, Germany.}}
\date{}

\maketitle

\begin{abstract}
We design a simple ascending-price algorithm to compute a $(1+\varepsilon)$-approximate equilibrium in Arrow-Debreu exchange markets with weak gross substitute (WGS) property, which runs in time polynomial in market parameters and $\log 1/\varepsilon$. This is the {\em first} polynomial-time algorithm for most of the known tractable classes of Arrow-Debreu markets, which is easy to implement and avoids heavy machinery such as the ellipsoid method. In addition, our algorithm can be applied in {\em unknown} market setting without exact knowledge about the number of agents, their individual utilities and endowments. Instead, our algorithm only relies on queries to a global demand oracle by posting prices and receiving aggregate demand for goods as feedback. When demands are real-valued functions of prices, the oracles can only return values of bounded precision based on real utility functions. Due to this more realistic assumption, precision and representation of prices and demands become a major technical challenge, and we develop new tools and insights that may be of independent interest.

Furthermore, our approach also gives the {\em first} polynomial-time algorithm to compute an exact equilibrium for markets with spending constraint utilities, a piecewise linear concave generalization of linear utilities. This resolves an open problem posed in~\cite{DuanM15}. 
\end{abstract}

\thispagestyle{empty}
\clearpage
\setcounter{page}{1}


\section{Introduction}
%
%

The concept of market equilibrium is central in economics and captures fair, stable, and efficient outcomes in competitive allocation scenarios. The most prominent model to study market equilibria are \emph{exchange markets}~\cite{ArrowD54}, which consist of a set of divisible goods and a set of agents. Each agent has an initial endowment of goods and a utility (preference) function over bundles of goods. Given prices of goods, each agent buys a most preferred bundle (called \emph{demand}) that is affordable from the earned money. At equilibrium, market clears, i.e., demand meets supply. 

The computation of market equilibrium is a fundamental problem in economics and computer
science~\cite{ArrowH58,Scarf73,Eaves76,Eaves85}. The challenge is to provide algorithms to compute an (approximate) market equilibrium efficiently. After more than a decade of research on this issue in theoretical computer science~\cite{DevanurPSV08,Jain07,CodenottiPV05,CodenottiMV05,CodenottiSVY06,DevanurK08,Orlin10,GargMV14,ChenDDT09,VaziraniY11,ChenPY13,DuanM15}, there is a fairly good understanding of tractable and intractable domains (assuming \classP\ $\neq$ \classPPAD). For exchange markets, the tractable case is essentially\footnote{The other tractable cases are CES utilities with $-1 \le \rho < 0$~\cite{CodenottiMPV05}, and piecewise linear concave utilities provided that the number of goods are constant~\cite{DevanurK08}.} given by markets with a {\em weak gross-substitutes} (WGS) property, where any increase in prices of a set of goods does not strictly decrease the demand of untouched goods. This includes markets with many popular and interesting classes of utility functions~\cite{MasColellWG95}, for example utilities with constant elasticity of substitution (CES) with $0 < \rho < 1$, Cobb-Douglas utilities, linear utilities, or utilities with spending constraints\footnote{Spending constraint utilities are a piecewise linear concave generalization of linear utilities, which satisfy many desirable properties and have many additional applications.}~\cite{Vazirani10,DevanurV04}. The only general approach to compute an approximate market equilibrium for this class of markets uses the ellipsoid method and convex feasibility programs~\cite{CodenottiPV05}. 

Algorithms based on the ellipsoid method, although being polynomial time in principle, have a number of undesirable drawbacks. First, the ellipsoid method often tends to be slow in practice. Moreover, such algorithms often are not informative about the problem structure since they invoke a powerful algorithmic tool in a ``black-box'' fashion. In many algorithmic domains it has become an important line of research to avoid black-box application of the ellipsoid method and derive simpler, ``combinatorial'' algorithms that reveal inherent structure. Thereby, such algorithms also avoid possible overhead that comes with the application of broad algorithmic hammers and tend to be much faster. Despite extensive research~\cite{Jain07,CodenottiMPV05,CodenottiPV05,CodenottiMV05}, the search for a simple and direct polynomial-time algorithm for computing market equilibria has so far been unsuccessful. The only exceptions are a classic strongly polynomial-time algorithm for the case of Cobb-Douglas utilities~\cite{Eaves85}, and recent combinatorial algorithms for linear utilities~\cite{DuanM15,DuanGM16}.

A prerequisite for all these efficient algorithms is the entire description of the market including number of agents, their utility functions and initial endowments. Obtaining this description is a highly non-trivial task; an entire theory of {\em revealed preferences}~\cite{Samuelson48,Varian82,Varian05} was developed to study how to infer market parameters from observed prices and buying patterns. Further, sometimes agents' preferences can be complicated, and their utility function may not have any succinct representation. This gives rise to the following question: Can we design efficient algorithms that are oblivious to market parameters? In other words, are there efficient algorithms for {\em unknown markets}?

Even beyond the economic interpretation, algorithms for computing market equilibria have wide applicability. For example, proportionally-fair allocations, which result from market equilibria, are widely used in the design of computer networks~\cite{KellyV02}. Recent applications also include energy-efficient scheduling~\cite{ImKM14,ImKM15} and fair allocation of indivisible resources~\cite{ColeGG13,ColeG15}. In the latter application, a combinatorial polynomial-time algorithm developed for linear markets~\cite{Orlin10} is extended to a problem, for which no convex programming formulation is known. This further highlights the importance of designing algorithms without the ellipsoid method that are simple and flexible to adjust to applications in related domains.

In this paper, we design simple ascending price algorithms for WGS exchange markets. Under the standard assumptions that the most preferred bundles are unique and continuous in all prices, our algorithm computes an $(1+\varepsilon)$-approximate market equilibrium in polynomial time, i.e., in time polynomial in market parameters and $\log(1/\varepsilon)$. In markets with linear and spending constraint utilities the uniqueness and continuity assumptions do not hold. Thus, even with tie-breaking rules that imply the WGS property, the existing algorithms for general WGS markets are not applicable. Instead, our algorithm can be adapted when using an appropriate tie-breaking rule for demands. We improve over the FPTAS~\cite{DevanurV04} in terms of running time from $1/\varepsilon^2$ to $\log(1/\varepsilon)$. More importantly, the $\log(1/\varepsilon)$-dependence of the running time allows us to convert the approximate equilibrium to an exact one in polynomial time when given all utility parameters. We obtain the first polynomial-time algorithm for computing an \emph{exact} market equilibrium for spending constraint utilities and settle an open question raised by \cite{DuanM15}. 

All our results are achieved by a unifying framework that works directly on the price vector. It uses simple binary search to identify suitable multiplicative price updates for subsets of goods. As such, our approach is extremely easy to implement and, in particular, avoids black-box use of the ellipsoid method. Furthermore, by working on the full generality of exchange markets, our algorithms can be applied to the simpler case of {\em Fisher markets}~\cite{BrainardS00}\footnote{Fisher markets are a subclass of exchange markets where buyers and sellers are different agents. Buyers bring money to buy goods while sellers bring goods to earn money.}.

Moreover, to compute approximate equilibria our algorithms do not require access to an explicit description of the utility functions and endowments of the agents. Instead, the algorithms even work in what we term \emph{unknown markets} -- the number of agents, the agents' endowments and utilities all remain unknown. As in the case of t\^atonnment algorithms~\cite{Uzawa60,Blad78,ColeF08,BirnbaumDX11,CheungCR12,CheungCD13,AvigdorRY14}, we assume these parameters can only be queried implicitly via aggregate demand queries. In such a query, we present a non-negative price for each good. For each agent, the oracle translates the endowment of each agent into money and determines a utility-maximizing bundle of goods for the money. It then aggregates the demands for each good and returns the vector of aggregate demands as answer to the query. Note that for many WGS markets these demand oracles can be easily implemented, e.g., for CES utilities there are even closed-form formulas. For suitable tie-breaking, we can even apply the algorithms to unknown markets with linear and spending constraint utilities, which have non-continuous demand. Perhaps surprisingly, our t\^atonnement-style algorithms succeed to detect implicitly all relevant information and changes in preferences in an unknown market setting. This requires to overcome a number of technical challenges and introduction of new techniques that we discuss in more detail below.

\subsection{Model and Notation}
\label{sec:model}

\paragraph{Exact and Approximate Market Equilibrium} 
In an exchange market there is a set $A$ of $n$ \emph{agents} and a set $G$ of $m$ \emph{goods}. Each agent $i$ has an initial endowment $w_{ij} \in
\R_{\ge 0}$ of good $j$. We denote the total endowment of good $j$ by $w_j = \sum_{i \in A} w_{ij}$ and assume w.l.o.g.\ that $w_j = 1, \forall j\in
G$. An \emph{allocation} $\vecx = (x_{ij})_{i \in A, j \in G}$ is an assignment of goods to agents such that $x_{ij} \ge 0, \forall i\in A, j\in G$, and $\sum_{i} x_{ij} = 1$. Each agent $i$ has a \emph{utility function} $u_i(\vecx_i)$ which specifies the value agent receives from his bundle $\vecx_i$. An (exact) \emph{market equilibrium} is a pair $(\vecx,\vecp)$, where $\vecx$ is an allocation and $\vecp = (p_j)_{j \in G}$ is a vector of non-negative \emph{prices} $p_j \ge 0$. In a market equilibrium, each agent obtains a budget of money by selling his endowment at the given prices. Then $\vecx_i$ represents a utility-maximizing bundle of goods that he can afford to buy at the current prices for his budget. We call such a bundle a \emph{demand bundle} of agent $i$. In addition, we say $x_j = \sum_i x_{ij}$ is the \emph{demand for good $j$} in allocation $\vecx$.

\begin{definition}
A bundle $\vecx_i$ is a \emph{demand} of agent $i$ at prices $\vecp$ if $u_i(\vecx_i) = \max\{ u_i(\vecy_i) \mid \vecy_i \in \R_{\ge 0}^m, \vecp^T
\vecy_i \le \vecp^T \vecw_i\}$.  A pair $(\vecx,\vecp)$ is a \emph{market equilibrium} if (1) $\vecx_i$ is a demand of agent $i$ at $\vecp$, and (2) $\sum_i x_{ij} = \sum_i w_{ij}$ for each good $j \in G$. If a pair $(\vecx,\vecp)$ is not a market equilibrium, the \emph{excess demand} $z_j = x_j - 1$ is non-zero for at least one good $j \in G$.
\end{definition}

Let us consider a concept of \emph{strong} approximate market equilibrium~\cite{CodenottiPV05}. It relaxes only the market clearing constraint but not the condition that $\vecx_i$ is a demand for each agent $i$.

\begin{definition}
A pair $(\vecp,\vecx)$ is a \emph{$\mu$-approximate equilibrium} $(\mu \ge 1)$ if (1) for each agent $i$, $\vecx_i$ is a demand of agent $i$ at prices $\vecp$, and (2) $\sum_i x_{ij} \le \mu \sum_i w_{ij}$ for each good $j$. 
\end{definition}

In Fisher markets~\cite{BrainardS00}, agents are divided into buyers and sellers. Each buyer $i$ comes with an initial endowment of money $B_i$. Each seller $i$ has an initial endowment of goods $\vecw_i$. Each buyer has no value for money and is only interested in buying goods, and each seller is only interested in obtaining money. Fisher markets are a special case of exchange markets when we interpret money as a separate good.

\paragraph{Utility Functions}
Algorithms for computing market equilibria rely on structural properties of the utility functions. A natural class are \emph{linear utilities} when each agent $i$ has non-negative values $u_{ij} \ge 0$ for each good $j \in G$, and $u_i(\vecx) = \sum_{j \in G} u_{ij} x_{ij}$. As a generalization, we consider spending constraint utilities~\cite{DevanurV04}, where the utility derived by agent $i$ from good $j$ is given by a piecewise linear concave (PLC) function $f_{ij}$. The overall utility of agent $i$ is additively separable among goods, i.e., $u_i(\vecx) = \sum_{j\in G} f_{ij}(x_{ij})$. Each $f_{ij}$ is a PLC function with a number of linear segments. Each segment $k$ has two parameters: the rate of utility $u_{ijk}$ per unit of good derived on segment $k$ and the maximum fraction $B_{ijk}$ of budget that can be spent on segment $k$. All $B_{ijk}$ are strictly positive, and concavity implies $u_{ijk} > u_{ij(k+1)}$. 
Here we assume all $u_{ijk}$'s are integers, all $B_{ijk}, w_{ij}$'s are rational numbers and the whole input can be represented in no more than
$L$ bits. Markets with spending constraint utilities have an equilibrium composed of rational numbers under mild sufficiency conditions (see Section \ref{sec:spending} for details).

More generally, we consider non-decreasing utility functions that generate markets with the \emph{weak gross-substitutes (WGS)} property -- when we increase a price, the demand for goods with untouched prices does not strictly decrease. For general WGS markets, we will assume that all demand bundles of agents are unique. Prominent examples are markets with \emph{Cobb-Douglas utilities} $u_i(\vecx) = \prod x_{ij}^{u_{ij}}$, or \emph{constant-elasticity-of-substitution (CES) utilities} $u_i(\vecx) = \left(\sum_{j \in G} u_{ij} x_{ij}^\rho\right)^{1/\rho}$ with $0 < \rho < 1$. Note that even if all utility and endowment parameters are rationals of finite size, then demand bundles and the market equilibrium in such markets might involve irrational numbers. In this case, we are interested in approximate market equilibria, and our prices will use a prespecified precision depending on the desired approximation factor.

\paragraph{Oracles}
Our algorithm queries demands for the agents by publishing prices $\vecp$. Then an oracle returns the total demand $x_j$ for each good $j \in G$. It assumes each agent can sell his initial endowment at the given prices and then requests a utility-maximizing bundle of goods for the money he has available. For ease of notation, given any price vector $\vecp$, let $\Oracle(\vecp)$ denote the {\it surplus vector} $\vecs = (s_1, \ldots, s_m)$ for the return of the oracle, where $s_j = p_j z_j$ is the \emph{surplus} (in terms of money) of good $j \in G$. In other words, assuming we publish $\vecp$, the oracle returns an excess demand vector $\vecz = (z_1, \ldots, z_m)$, then\footnote{We define the surplus vector as price times excess demand since it satisfies $\sum_i{s_i} = 0$. This invariant is useful in design and analysis of our algorithms below.} $\Oracle(\vecp) = \vecp \cdot \vecz$.

In general WGS markets the surplus vector might contain irrational values. Thus, we use an {\it approximate demand oracle} $\AOracle(\vecp, \mu)$, which is a blackbox algorithm that takes any price vector $\vecp$ and positive rational $\mu$ as input. It returns a surplus vector $\vecs$ such that $|s_i - \Oracle(\vecp)| \leq \mu$ holds for every good $i$. Note that for many WGS markets the demand oracle can be implemented very quickly -- for CES utilities there are even closed-form expressions for the demand of each agent for each good as a function of prices, utility and endowment. More generally, we assume that the oracle can be implemented in time polynomial in the input size and $\log(1/\mu)$. This standard assumption for demand oracles has also been used in previous work~\cite{CodenottiPV05,CodenottiMV05}. 

In linear and spending constraint markets, a major challenge for an algorithm in unknown markets is non-uniqueness of demands. Here an oracle needs to do tie-breaking between several different bundles of goods that yield the same maximum utility for an agent. Ideally, it should satisfy the following properties: (1) The output demand is always deterministic and unique. (2) If $\vecp$ are equilibrium prices, the output demand equals supply for every good. (3) The oracle can be implemented in time polynomial in the input size. Based on these criteria, we use a demand oracle that yields demands minimizing the $\ell_2$-norm of the surplus vector. More formally, for prices $\vecp$ our oracle returns a set of demand bundles for the agents such that $\sum_i s_i^2$ is minimized, where $s_i$ is the surplus of good $i$. Such a tie-breaking rule was introduced by~\cite{DevanurPSV08} and it has been later used in several market algorithms, such as~\cite{DuanGM16, DuanM15, Vazirani10}. This oracle satisfies all three properties. 
Furthermore, if utilities, endowments and prices are all given as integers with a number of bits polynomial in $m$ and $L$, we can represent every surplus $s_i$ also with a number of bits polynomial in $m$ and $L$. Hence by setting $\mu$ sufficiently small, we can convert an approximate demand oracle $\AOracle$ into an exact oracle in polynomial time. Therefore we will assume that in spending constraint markets we are equipped with an exact demand oracle $\Oracle$ instead of an approximate one. 

\subsection{Results and Contribution}
We present simple ascending-price algorithms that converge to market equilibrium in WGS and spending constraint markets. Our algorithm for WGS markets in Section~\ref{sec:wgs} converges to a $(1+\varepsilon)$-approximate equilibrium in time polynomial in $m$, $\log(1/\varepsilon)$, and other market parameters. We present the first algorithm which is not based on the ellipsoid method for this general class of markets. Furthermore, the number of agents, the agents' endowment, and utilities all remain unknown. We only query a global demand oracle that provides aggregate demands for goods at given prices. This information is then used to increase prices of a selected set of goods whose demand is more than their supply. In Section~\ref{sec:wgs} we ignore for simplicity all precision and representation issues to highlight the general proof technique. The complete analysis with all details of our algorithm is provided in Appendix~\ref{app:wgs}, where we specify in advance a precision for prices and rely on approximate demand oracles whose output is within our desired bit precision. One can view our algorithm as a form of t\^atonnement. It improves over previous approaches~\cite{CodenottiMV05} by exponentially decreasing the dependence of the running time on $1/\varepsilon$ and by working directly with the unknown market without transformations like adding auxiliary agents.

Next, we apply our approach in Section~\ref{sec:spending} to spending constraint utilities -- a piecewise linear concave generalization of linear utilities, which has many additional applications due to its natural diminishing returns property. Since these markets have demands that are non-continuous in the prices, we cannot directly apply our algorithm or other previous algorithms for WGS markets. Instead, we adjust our approach to implicitly capture the non-continuity events when using a global demand oracle with suitable tie-breaking. When all parameters are represented by at most $L$ bits, our algorithm computes even an exact market equilibrium in time polynomial in $m$ and $L$. All prices and demands occurring during the algorithm require a bit precision polynomial in $m$ and $L$. It first computes a $(1+\varepsilon)$-approximate equilibrium using a precision that is polynomial in $m$, $L$ and $\log(1/\varepsilon)$. The exact demands returned by the demand oracle have the same precision. For a small enough $\varepsilon$ (using only polynomial bit length), we can then use a rounding procedure to turn it into a price vector of an exact market equilibrium. 

Note that our algorithm requires only access to a suitable demand oracle to compute an approximate equilibrium. However, in contrast to WGS markets, for spending constraint markets the oracle uses global tie-breaking. To obtain an exact equilibrium, our final rounding procedure relies on full information about the utilities. Thus, we obtain in polynomial time an approximate equilibrium in unknown markets (with global tie-breaking) and an exact equilibrium with full information. This represents the first polynomial-time algorithm to compute an exact market equilibrium for spending constraint utilities and settles an open question raised by~\cite{DuanM15}. An important open problem is to construct efficient t\^atonnement algorithms that avoid global tie-breaking and full-information rounding.

We use and extend ideas of algorithms for linear markets with full information~\cite{DevanurPSV08,DuanM15}. These ideas were also used in spending constraint markets~\cite{DevanurV04} to compute a $(1+\varepsilon)$-approximate equilibrium in time polynomial in $m$, $L$ and $1/\varepsilon$. Roughly speaking, they are intimately tied to linear and spending constraint utilities, where they work on the agents' side and increase prices until structural changes occur in the optimal bang-per-buck relations. The progress towards equilibrium is measured in the reduction of the $\ell_2$-norm of surplus. Our approach works for all WGS markets, in which in general do not exhibit such structural events that can be used for analysis. This turns out to be much more demanding, and we develop an approach that works entirely on the goods' side, based only on prices and aggregate demands obtained via oracle access.

An issue that plays a central role in our algorithms for spending constraint utilities is precision of prices and demands. This seems to have been treated only in minor detail in some of the previous works. For the algorithm of~\cite{DuanM15} in linear exchange markets these issues are discussed in depth. However, their solution is to change the agents side and alter utility values for maintaining bounded precision throughout the algorithm. However, computing an approximate equilibrium in unknown markets with demand oracle access seems impossible via this route.

As a consequence, unlike for the existing algorithms in the linear case, the surpluses encountered by our algorithms might now become negative. Hence, additional events have to be taken into account upon increasing prices. Moreover, a significant challenge lies in maintaining the precision of prices to be polynomial throughout the algorithm. To overcome this problem, we make use of a novel tool that we term \emph{ratio graph}. This graph is defined for a vector of prices $\vecp$. The goods are the vertices, and we draw an undirected edge between goods $j$ and $k$ if the ratio of prices $p_j/p_k$ can be expressed by two $L$-bit numbers. For an intuition, observe that if some agent $i$ has the same bang-per-buck for two goods $j$ and $k$, then $u_{ij}/p_j = u_{ik}/p_k$ or $p_j/p_k = u_{ij}/u_{ik}$, i.e., the ratio of prices can be expressed by two $L$-bit numbers. Maybe surprisingly, the broad structure of the ratio graph indeed contains enough information to implement algorithms for finding approximate equilibria in unknown spending constraint markets.

\subsection{Related Work}
\label{sec:relatedWork}

The problem of computing market equilibria has been intensively studied, and the literature is too vast to survey here. We provide an overview of the
work most directly relevant to ours. There is a large body of work on algorithms for computing equilibrium using full market information. The
first combinatorial polynomial-time algorithm for linear Fisher markets was given by~\cite{DevanurPSV08}. Later,~\cite{Vazirani10} provided a polynomial-time algorithm for Fisher markets with spending constraint utilities by extending combinatorial techniques of~\cite{DevanurPSV08}. Strongly polynomial-time algorithms are also known, for Fisher markets with linear~\cite{Orlin10,Vegh12} and spending constraint utilities~\cite{Vegh12}.

For linear exchange markets,~\cite{Jain07} and~\cite{Ye07} obtained polynomial-time algorithms based on ellipsoid and interior point methods on a convex program, respectively.~\cite{DuanM15} gave the first combinatorial polynomial-time algorithm for this problem, which was recently improved by~\cite{DuanGM16}. For exchange markets with spending constraint utilities,~\cite{DevanurV04} gave an algorithm to compute a $(1+\epsilon)$-approximate equilibrium, for which the running time dependence on $\epsilon$ is $O(1/\epsilon^2)$.~\cite{Eaves85} gave a strongly polynomial-time algorithm for markets with Cobb-Douglas utilities.

For the general case of WGS markets with unique and continuous demands, a polynomial-time algorithm was obtained by~\cite{CodenottiPV05}. Note that this algorithm relies heavily on the ellipsoid method. For the Fisher setting, the famous Eisenberg-Gale convex program~\cite{EisenbergG59} captures market equilibrium under linear utilities.~\cite{Eisenberg61} generalized it to work for any homogeneous utility functions, many of which satisfy the WGS property.


As an alternative approach to compute market equilibrium, {\em t\^atonnement} was defined by~\cite{Walras74} -- algorithms which have access to endowment and utilities only via demand oracles. Usually, t\^atonnement procedures conduct price updates separately for each good, sometimes based on derivatives of the demand as a function of price. 
In the computer science literature,~\cite{CodenottiMV05} gave a discrete t\^atonnement process that converges to an $(1+\epsilon)$-approximate equilibrium for WGS markets. It has a convergence time polynomial in the input size and $1/\epsilon$, and it does not query the original market since one needs to add auxiliary agents. By taking a different approach, our algorithm for the same market setting improves this rate to polynomial in the input size and $\log(1/\epsilon)$. Also, we only rely on approximate demand queries to the original market.

More recently,~\cite{ColeF08} established the first fast converging discrete version of t\^atonnement for WGS markets. The convergence time depends on various market parameters. It requires a non-zero amount of money in the market, so it works for the special case of Fisher markets and beyond, but it is not applicable to the full range of exchange markets. For Fisher markets, many additional results~\cite{BirnbaumDX11,CheungCR12,CheungCD13,CheungC14,AvigdorRY14} on the convergence of t\^atonnement processes beyond WGS markets were derived.





\newcommand{\AlgExchange}{{\textsc{Alg-WGS}}\xspace}
\newcommand{\AlgExchangePrecise}{{\textsc{Alg-WGS-Precise}}\xspace}
\section{WGS Exchange Markets}
\label{sec:wgs}

In this section we describe the algorithm for WGS exchange markets. As previously mentioned, we assume that we are only granted access to an approximate demand oracle and are restricted to finite precision arithmetic computations. In order to make our algorithm and its analysis more accessible, we will simplify the problem in the remainder of this section by assuming that we are equipped with (1) exact real arithmetic, and (2) an exact demand oracle. This significantly simplifies the analysis in terms of notation and calculations, and as such concentrates on the key ideas of the algorithm. A complete and rigorous version of this section, presenting the entire algorithm and its proof with approximate precision, can be found in Appendix~\ref{app:wgs}.

We apply algorithm \AlgExchangePrecise\ given below. The main idea is to repeatedly identify a subset of goods $G_1$ by finding a gap in the sorted order of surpluses. It then raises the prices of $G_1$ by a common factor $x$ until the surplus gap is almost closed, or the smallest surplus of $G_1$ becomes (close to) 0. More formally, given price vector $\vecp=(p_1, \ldots, p_m)$, value $x \in \R^+$ and subset $S \subseteq G$, the algorithm uses $\UpdatePrice(\vecp, x, S)$, which is the price vector $\vecp' = (p'_1, \ldots, p'_m)$ with $p'_i = x\cdot p_i$ if $i \in S$ and $p'_i = p_i$ otherwise. To implement this process, the algorithm relies on two parameters $D_1$ and $D_2$ based on the following assumptions.

\begin{assumption}[Bounded Price]\label{asp1}
  There exists a market equilibrium $(\vecp^*,\vecx^*)$ with $1 \le p_i^* \le 2^{D_1}, \forall i \in G$.
\end{assumption}

\begin{assumption}[Continuity]\label{asp2}
  For any price vector $\vecp$ such that $1 \leq p_i \leq 2^{D_1}$ for each $i$, $|\frac{\partial s_i}{\partial p_j}| < 2^{D_2}$ for every $i, j$, where $s_i$ is the surplus money of good $i$ in $\Oracle(\vecp)$, and $D_2$ is a polynomial of the input size.
\end{assumption}

These two assumptions are precisely the ones from~\cite{CodenottiMV05,CodenottiPV05}. Assumption~\ref{asp1} about bounded prices is fairly mild and in many cases necessary for an efficient algorithm to compute a (strong) approximate market equilibrium. Assumption~\ref{asp2} about continuity is also satisfied by many natural markets, for example, markets with CES utilities with $0 < \rho < 1$. Note, however, that it is not satisfied for linear and spending constraint markets, and hence we must develop new tools and procedures in Section~\ref{sec:spending}.

\begin{algorithm}[t]
\caption{\AlgExchangePrecise\label{alg:WGS-Precise}}
\DontPrintSemicolon
\SetKwInOut{Input}{Input}\SetKwInOut{Output}{Output}
\SetKw{Param}{Parameters:}
\Input{number of goods $m$, demand oracle $\Oracle$, precision bound $\varepsilon > 0$}
\Output{Prices $\vecp$ of an $(1+\varepsilon)$-approximate market equilibrium}

\BlankLine
Set initial price $\vecp_0 \ot (1, 1, \ldots, 1)$ and round index $t \ot 0$.\;
\Repeat( \ //round $t$){$\|\Oracle(\vecp_t)\|^2 < \varepsilon^2$}{
$t \ot t+1$\;
$\vecs = (s_1, \ldots, s_m) \ot \Oracle(\vecp_{t-1})$\;
Sort $\vecs$ such that $s_{i_1} \geq s_{i_2} \geq \cdots \geq s_{i_m}$.\;
Find smallest $k$, such that $s_{i_{k+1}} \leq 0$ or $s_{i_k} >
(1+\frac1m)s_{i_{k+1}}$. \;
Set $G_1 \ot \{i_1, \ldots, i_k\}$ and $G_2 \ot G \setminus G_1$\;
Find the largest $x$, such that in $\vecs' = \Oracle(\UpdatePrice(\vecp_{t-1}, x, G_1))$ it holds\\ \nonl \hspace{0.5cm} $\min\{s'_i \mid i \in G_1\} = \max\{\{s'_i \mid i \in G_2\}\cup \{0\}\}$.\;
$\vecp_{t} \ot \UpdatePrice(\vecp_{t-1}, x, G_1)$\;
}
\Return{$\vecp_t$}
\end{algorithm}

To measure progress towards equilibrium we use a \emph{potential function} $\Phi(\vecp_t) = \|\Oracle(\vecp_t)\|^2$. 
We start by proving a number of claims about the price vector $\vecp_t$. The first claim shows that with respect to exact demands, our algorithm monotonically reduces the 1-norm of the surpluses of all goods from $2m$.

\begin{claim}\label{claim:surplusBound}
  In \AlgExchangePrecise, $|\Oracle(\vecp_0)| \leq 2m$ and $|\Oracle(\vecp_t)|$ is non-increasing in $t$.
\end{claim}

\begin{proof}
  Let $d_i$ be the exact demand for good $i$ under price $\vecp_0$, then $|\Oracle(\vecp_0)| = \sum_i|d_i - 1| \leq \sum_id_i + m = 2m$. Next, by the criteria to define $G_1$ and $G_2$ in each round, we have $\{i \mid \Oracle(\vecp_{t-1})_i < 0\} \subseteq G_2$.
  
  During round $t$, only prices of goods in $G_1$ are increased. By the WGS property, we know $\Oracle(\vecp_{t})_i \geq \Oracle(\vecp_{t-1})_i$ for every $i \in G_2$. Further, note that $\min\{\Oracle(\vecp_t)_i \mid i \in G_1\} \geq 0$ since $\min\{\Oracle(\vecp_t)_i \mid i \in G_1\} \geq 0$. Hence, we do not introduce any new negative surplus in $\Oracle(\vecp_{t})$. Thus, we have 
{\small
$$|\Oracle(\vecp_{t-1})| = -2\sum_{\Oracle(\vecp_{t-1})_i < 0}\Oracle(\vecp_{t-1})_i \geq -2\sum_{\Oracle(\vecp_{t})_i < 0}\Oracle(\vecp_{t-1})_i \geq
-2\sum_{\Oracle(\vecp_{t})_i < 0}\Oracle(\vecp_{t})_i = |\Oracle(\vecp_{t})|\enspace.$$}
\end{proof}

The next two claims bound the range of prices we encounter, which is important for showing that we approach the unique market equilibrium. 

\begin{claim}\label{claim:price1}
  Throughout the run of \AlgExchangePrecise, every good with negative surplus has price 1. Hence, there will be at least one good whose price remains 1.
\end{claim}
\begin{proof}
  Observe the following three simple facts about the surplus $\Oracle(\vecp_t)$ resulting from exact demands: (1) Throughout the algorithm we never increase the price of any good with negative surplus. (2) The surplus of any good does not change from non-negative to negative. (3) For any non-equilibrium price vector, there will always be a good with negative surplus. These facts are direct consequences of the conditions used to classify goods based on $\Oracle(\vecp_t)$ in the algorithm. Together they prove the claim.
\end{proof}

\begin{claim}\label{claim:priceUpper}
  In \AlgExchangePrecise, for any $t \geq 0$, all prices in $\vecp_t$ are bounded by $2^{D_1}$.
\end{claim}
\begin{proof}
  Let $\vecp^*$ be equilibrium prices according to Assumption~\ref{asp1}. We show that for any $t \geq 0$, $\vecp_t$ is always pointwise smaller than $\vecp^*$. Assume that this is not true, let $t$ be the smallest value such that there exists $(\vecp_{t})_i > \vecp^*_i$ for some $i$. Note that according to the algorithm, we have $\vecp_t = \UpdatePrice(\vecp_{t-1}, x, G_1)$ for some $x > 1$ and $G_1 \subseteq [m]$. Further, by our classification based on $\Oracle$, it is easy to see that $\Oracle(\vecp_t)_i > 0$ for any $i \in G_1$. This means from $\vecp_{t-1}$ to $\vecp_t$, only prices of goods in $G_1$ are increased. Let $S = \{i \mid (\vecp_t)_i > \vecp^*_i\}$, then we have $S \subseteq G_1$.

  Next, we apply a sequence of price changes to $\vecp_t$. First, for every $j \notin S$ we increase $(\vecp_t)_j$ to $\vecp^*_j$. Let $\vecp$ be the new price vector and consider the surplus $\Oracle(\vecp)$ resulting from exact demands. By the WGS property of the market, the surplus of any good in $S$ will not decrease, hence we still have $\Oracle(\vecp)_i > 0$ for every $i \in S$. The sum of all surpluses in an exchange market is always 0, so $\sum_{j \notin S}\Oracle(\vecp)_j < 0$.

  Now we decrease the price of every good $i \in S$ from $\vecp_i$ to $\vecp^*_i$. Then $\vecp$ becomes exactly $\vecp^*$. This process will not increase surplus of any good $j \not\in S$. Thus we still have $\sum_{j \notin S}\Oracle(\vecp^*)_j < 0$. This contradicts the assumption that $\vecp^*$ are prices of a market equilibrium.
\end{proof}

The following claim establishes a relation between the surplus of a good with respect to exact demands before and after a multiplicative price update step.

\begin{claim} \label{claim:xp}
  For any price vector $\vecp$, $x > 1$ and $S \subseteq [m]$, $\Oracle(\UpdatePrice(\vecp, x, S))_i \leq x \cdot \Oracle(\vecp)_i$ for any $i \in S$.
\end{claim}

\begin{proof}
  Because prices are scalable in exchange markets, we have $\Oracle(x\cdot \vecp) = x\cdot \Oracle(\vecp)$ for any value $x > 0$. Also, by the WGS property, when we decrease any set of prices, this will not increase the demand for goods with untouched price. Since these goods have untouched price and non-decreasing demand, they also enjoy non-decreasing surplus. Therefore for any $i \in S$, we have $\Oracle(\UpdatePrice(\vecp, x, S))_i \leq \Oracle(x\cdot \vecp)_i = x\cdot\Oracle(\vecp)_i$. 
\end{proof}

The next lemma is the key step in the proof of our main result. It establishes a multiplicative decrease of the potential function at the end of many of the rounds. Let $R$ be a sufficiently large constant to be explained in the end of the proof of the following lemma.

\begin{lemma}\label{lemma:potential}
  If $x < 1+\frac{1}{Rm^3}$ at the end of round $t$ in \AlgExchangePrecise, 
  then $\Phi(\vecp_t) \leq \Phi(\vecp_{t-1}) \left(1-\Omega\left(\frac{1}{m^3}\right)\right)$.
\end{lemma}

\begin{proof}
  We let $\vecs = \Oracle(\vecp_{t-1})$ and $\vecs' = \Oracle(\vecp_t)$ throughout the proof.
	An intuition of the proof is as follows. In \AlgExchangePrecise, by the conditions used to define $G_1$ and $G_2$, we always have $s_{i_k} \geq s_{i_1} /e$ and $s_{i_k} - s_{i_{k+1}} > s_{i_k}/(m+1) \geq s_{i_1} / e(m+1)$. Hence, roughly speaking, every good in $G_1$ has reasonably large surplus, and there is a reasonably large gap between the surpluses in $G_1$ and $G_2$. Next, at the end of the current round, we decreased the minimum surplus of a good in $G_1$ to either $\min\{s'_i \mid i \in G_1 \} = 0$ (Case (1) below) or $\min\{s'_i \mid i \in G_1 \} = \max\{s'_i \mid i \in G_2\}$ (Case (2) below). In both cases, the total value of $\Phi$ must decrease by a factor of $1 - \Omega(1/m^3)$.

  More formally, if the algorithm proceeds to round $t$, then $\|\vecs\| > \varepsilon^2$. By the definition of set $G_1$, we have $s_{i_1} \leq (1+\frac1m)s_{i_2} \leq \cdots \leq (1+\frac1m)^{k-1}s_{i_k} < e\cdot s_{i_k}$. Hence, $s^2_{i_k} > (s_{i_1}/e)^2 \geq \Phi(\vecp_{t-1})/(me^2) > (\varepsilon/(\sqrt{m}e))^2$, so the surpluses of goods in $G_1$ are similar up to a factor of $e$ and bounded from below. Also we have $(s_{i_k} - s_{i_{k+1}})^2 > (s_{i_1} / e(m+1))^2 \geq \Phi(\vecp_{t-1})/(e^2(m+1)^2m)$.

Next we relate the surplus in the beginning and the end of a round as follows. For every $i \in G_1$, by Claim~\ref{claim:xp}, the surplus from exact demands satisfies $s'_i \leq x\cdot s_i$. Since $x < 1+\frac1{Rm^3}$, it holds that $s'_i < (1+\frac{1}{Rm^3})s_i$. We do not touch the price of any good $j \in G_2$, so the WGS property implies for exact demands $s'_j \geq s_j$. 
  
Now, in order to bound the change of $\Phi(\vecp_t)$, we consider $\vecs'$ according to $G_1$ and $G_2$. We distinguish two cases.

\paragraph{Case 1: $\max\{s'_i \mid i \in G_2\} < 0$} 
In this case the algorithm has decreased the surplus of some good in $G_1$ to 0. This decrease alone brings down the potential function $\Phi$ by a factor of $1 - \Omega(1/m)$. All other surpluses will cause an increase by a factor of at most $1+O(1/m^3)$.

More formally, the contribution of goods of $G_1$ to $\Phi(\vecp_t)$ can be upper bounded by
{\small
$$\sum_{j=1}^ks'^2_{i_j} < \sum_{j=1}^{k-1}\left(1+\frac1{Rm^3}\right)^2s^2_{i_j}.$$}
Furthermore, for every $i \in G_2$, by the WGS property of the market, we know $s_i \leq s'_i < 0$. Thus the contribution of goods of $G_2$ to $\Phi(\vecp_t)$ can be upper bounded by
%
%
{\small
\begin{eqnarray*}
 \sum_{j=k+1}^m s'^2_{i_j}& \leq & \sum_{j=k+1}^m s^2_{i_j} < \sum_{j=k+1}^m \left(1+\frac1{Rm^3}\right)^2 s^2_{i_j} 
\end{eqnarray*}}
Combining the two parts 
{\small
\begin{eqnarray*}
	\Phi(\vecp_t) = \sum_{j=1}^ms'^2_{i_j} & < & \sum_{j\neq k}\left(1+\frac1{Rm^3}\right)^2s^2_{i_j}\\
    &=&\left(1+\frac1{Rm^3}\right)^2(\Phi(\vecp_{t-1})-s^2_{i_k})\\
    &<&\left(1+\frac1{Rm^3}\right)^2\left(1-\frac{1}{e^2m}\right)\Phi(\vecp_{t-1})\\
    &<&\left(1-\frac1{2e^2m}\right)\Phi(\vecp_{t-1})
\end{eqnarray*}}
where the last inequality holds for any $m \ge 2$ with sufficiently large constant $R$.

\paragraph{Case 2: $\max\{s'_i \mid i \in G_2\} \ge 0$} 
In this case the gap between surpluses in $G_1$ and $G_2$ decreases to 0. Below we show that the closing of this gap yields a decrease of the potential function $\Phi$ by a factor of $1 - \Omega(1/m^3)$. All other surpluses will increase by a factor of at most $1+O(1/m^3)$. In combination, it turns out that $\Phi$ will decrease by a factor of $1 - \Omega(1/m^3)$.

More formally, in this case $\min\{s'_i \mid i \in G_1\} = \max\{s'_j \mid j \in G_2\}$. Let $s_{G_1} = \min\{s'_i \mid i \in G_1\}$ and $s_{G_2}=\max\{s'_j \mid j \in G_2\}$.
For every $i \in G_1$, let $s'_i = x's_i - \delta_i$ where $x'=(1+\frac1{Rm^3})$, and for every $j \in G_2$, let $s'_j = s_j + \delta_j$. Hence $\delta_i, \delta_j \geq 0$ for all $i, j$. Further, we have $\sum_{i=1}^m s_i = \sum_{i=1}^m s'_i = 0$, hence
\begin{eqnarray*}
\sum_{i \in G_1}\delta_i & = & \sum_{i \in G_1} x's_i - \sum_{i \in G_1}s'_i \\
& \geq & \sum_{i \in G_1} s_i + \sum_{j \in G_2}s'_j = \sum_{i \in G_1} s_i + \sum_{j \in G_2} (s_j + \delta_j) \\
& = & \sum_{j \in G_2}\delta_j
\end{eqnarray*}
and 
\vspace{-0.4cm}
$$\sum_{i\in G_1}\delta_i \geq \frac12(s_{i_k}-s_{i_{k+1}}).$$
Now we have 
{\small 
\begin{eqnarray}
	\Phi(\vecp_t) & = & \sum_is'^2_i = \sum_{i \in G_1}(x's_i - \delta_i)^2 + \sum_{j \in G_2}(s_j + \delta_j)^2 \nonumber\\
    & = & \left(\sum_{i \in G_1}x'^2s^2_i + \sum_{j \in G_2}s_j^2\right) + \left(\sum_{j \in G_2}\delta_j(s_j+\delta_j)-\sum_{i \in G_1}\delta_i(x's_i-\delta_i)\right) - \sum_{i \in G_1}x's_i\delta_i + \sum_{j \in G_2}\delta_js_j \label{eq:line1}\\
    & < & x'^2\Phi(\vecp_{t-1}) + \left(s_{G_2}\sum_{j \in G_2}\delta_j - s_{G_1}\sum_{i \in G_1}\delta_i\right) - s_{i_k}\sum_{i \in G_1}\delta_i + s_{i_{k+1}}\sum_{j \in G_2}\delta_j \label{eq:line2}\\   %
    & < & x'^2\Phi(\vecp_{t-1}) - (s_{i_k} - s_{i_{k+1}})\sum_{i \in G_1}\delta_i\label{eq:line3}\\
    %
    %
    %
    %
    & < & x'^2\Phi(\vecp_{t-1}) - \frac{1}{2}(s_{i_k}-s_{i_{k+1}})^2 \label{eq:line5}\\
    & < & \left(1+\frac2{Rm^3}+\frac1{R^2m^6}-\frac1{2e^2(m+1)^2m}\right)\Phi(\vecp_{t-1}) \label{eq:line6}\\
    & = & \left(1-\Omega\left(\frac1{m^3}\right)\right)\Phi(\vecp_{t-1}) \label{eq:line7}
\end{eqnarray}
}
Here \eqref{eq:line1} can be derived by expanding the quadratic formula and appropriately reorganizing the terms. For the step from \eqref{eq:line1} to \eqref{eq:line2}, in the first bracket we overestimate the quadratic terms of $s$ into $x'^2\Phi(\vecp_{t-1})$. In the second bracket, the terms are bounded correctly using $s_{G_1}$ for all $i \in G_1$ and $s_{G_2}$ for all $j \in G_2$. For the final two terms in~\eqref{eq:line1} and \eqref{eq:line2} we use the definition of $s_{i_k}$ and $s_{i_{k+1}}$ and the fact that $x' > 1$. For the step from \eqref{eq:line2} to \eqref{eq:line3}, for the second bracket of \eqref{eq:line2} we note $s_{G_1} \ge s_{G_2}$ and the two sets $G_1$ and $G_2$ have the same sums of $\delta$-terms. By the same argument, we can transform the last two terms of \eqref{eq:line2} as shown. From \eqref{eq:line3} to \eqref{eq:line5} and then to \eqref{eq:line6}, we use the bound for $\sum_{i \in G_1}\delta_i$ and $(s_{i_k} - s_{i_{k+1}})^2 > \Phi(\vecp_{t-1})/(e^2(m+1)^2m)$ as shown above.

Finally, using a sufficiently large constant $R$, the multiplicative term in \eqref{eq:line6} can be decreased to strictly less than 1 for every $m \ge 2$. The final expression in~\eqref{eq:line7} captures the asymptotics and proves the lemma.
\end{proof}

Observe that the previous lemma shows a decrease in the potential only for rounds in which the factor $x$ is rather small. The next lemma shows that there can be only a limited number of rounds with a larger value of $x$.

\begin{lemma}\label{lemma:xmax}
  During a run of \AlgExchangePrecise, there can be only $O(m^4 D_1)$ many rounds that end with $x \geq 1+\frac1{Rm^3}$.
\end{lemma}
\begin{proof}
  By Claim~\ref{claim:priceUpper}, every price can be increased by a factor of $1+\frac1{Rm^3}$ at most $O(\log_{1+1/Rm^3}2^{D_1}) = O(m^3D_1)$ times. 
  Hence there can be at most $O(m^4D_1)$ many rounds with $x \geq 1+\frac1{Rm^3}$.
\end{proof}

Finally, we can assemble the properties to show the number of rounds to reach an $(1+\varepsilon)$-approximate market equilibrium is polynomially bounded.

\begin{lemma}\label{lem:exchangeWGS}
  For any market that satisfies Assumptions~\ref{asp1} and~\ref{asp2}, and for any $\varepsilon > 0$, \AlgExchangePrecise returns the price vector of an $(1+\varepsilon)$-approximate market equilibrium in a number of rounds polynomial in the input size and $\log(1/\varepsilon)$.
\end{lemma}

\begin{proof}
  Let $x_t$ be the value of $x$ we find in round $t$ of \AlgExchangePrecise. First because at least one price will increase by a factor of $x_t$ in round $t$, by Claim~\ref{claim:priceUpper} we have $\prod_tx_t \leq 2^{mD_1}$.
  At the end of round $t$, if $x_t \geq 1+\frac1{Rm^3}$, let $s=\max\{s_i \mid s_i \in \Oracle(\vecp_{t-1})\}$ 
  and $s'=\max\{s_i \mid s_i \in \Oracle(\vecp_{t})\}$, 
  then by Claim~\ref{claim:xp} we have $\Phi(\vecp_t) \leq ms'^2 \leq mx_t^2s^2 \leq mx_t^2\Phi(\vecp_{t-1})$.
  Moreover, by Lemma~\ref{lemma:xmax} there will be at most $O(m^4D_1)$ such rounds.
  Hence the total increase of $\Phi(\vecp_t)$ in these rounds will be no more than a factor of 
  $\prod_{x_t \geq 1+1/Rm^3}{mx_t^2} \leq m^{O(m^4D_1)}2^{2mD_1} = m^{O(m^4D_1)}$.

  For all other rounds, we have $x < 1+\frac1{Rm^3}$, and by Lemma~\ref{lemma:potential}, the potential function is decreased by a factor of $1/(1-\Omega(\frac1{m^3}))$. Therefore the total number of rounds before $\Phi(\vecp_t) \leq \varepsilon^2$ will be at most 
{\small
  $$O\left(\log_{1/(1-\Omega(\frac1{m^3}))}\frac{m^{O(m^4D_1)}}{\varepsilon^2}\right) = O\left(D_1m^7\log{m}+m^3\log{\frac1\varepsilon}\right)$$
  }
%
\end{proof}

The final step is to argue that the total running time (in terms of oracle queries) of the algorithm is polynomial. Given Lemma~\ref{lem:exchangeWGS}, it remains to show that each round only invokes polynomially many oracle queries. Intuitively, this can be achieved by finding $x$ in line 8 of \AlgExchangePrecise via binary search. However, for a formal proof it is unavoidable to make statements about the required bit precision of $\vecp$, $x$, and the approximate demand oracle, which we did not discuss here. In the following we present the statement of the main theorem. The formal proof, together with the statement of algorithm \AlgExchange based on bounded precision and full details on its analysis, is deferred to Appendix~\ref{app:wgs}.
\begin{theorem}\label{thm:exchangeWGS}
  For any market that satisfies Assumptions~\ref{asp1} and~\ref{asp2}, and for any $\varepsilon > 0$, \AlgExchange returns the price vector of an $(1+\varepsilon)$-approximate market equilibrium in time polynomial in the input size and $\log(1/\varepsilon)$.
\end{theorem}

\begin{remark} \rm
Our main goal in the analysis was to establish a bound on the running time that is polynomial in $m$, $L$, and $\log(1/\varepsilon)$. For the sake of simplicity we did not optimize the bounds beyond being polynomial. It appears that the dependence on $m$ can be significantly improved, e.g., by a more precise analysis of the actual number of rounds with $x \geq 1+\frac1{Rm^3}$ and their impact on the potential. Moreover, based on our preliminary experiments, it appears that using $\mu = \Theta(\varepsilon/m^4)$ is sufficient, and the algorithm converges to an equilibrium much faster than the bound predicts.
\end{remark}

\newcommand{\AlgSC}{{\textsc{Alg-Spending}}\xspace}
\newcommand{\Rounding}{{\textsc{Rounding}}\xspace}
\newcommand{\AlgSCRounding}{{\textsc{Alg-Spending-Rounding}}\xspace}
\newcommand{\AlgSCExact}{{\textsc{Alg-Spending-Exact}}\xspace}

\section{Exchange Markets with Spending Constraint Utilities} \label{sec:spending}

In this section we discuss our algorithm for exchange markets with spending constraint utilities. Spending constraint utilities are defined in~\cite{Vazirani10,DevanurV04}, where the utility derived by agent $i$ from good $j$ is given by a piecewise linear concave (PLC) function $f_{ij}$. The overall utility of agent $i$ is additively separable among goods, i.e., $u_i(\vecx) = \sum_{j\in G} f_{ij}(x_{ij})$. Each $f_{ij}$ is a PLC function with a number of linear segments. Each segment $k$ has two parameters: the rate of utility $u_{ijk}$ per unit of good derived on segment $k$ and the maximum fraction $B_{ijk}$ of budget that can be spent on segment $k$. All $B_{ijk}$ are strictly positive, and concavity implies $u_{ijk} > u_{ij(k+1)}$. Here we assume all $u_{ijk}$'s are integers, all $B_{ijk}, w_{ij}$'s are rational numbers and the whole input can be represented in no more than $L$ bits. 

Spending constraint markets may not have an equilibrium \cite{DevanurGV13}, however under mild conditions, there is always a rational equilibrium~\cite{Maxfield97}. Henceforth we will assume the following sufficient condition. Let $\Gamma(S) = \{j\in G\ |\ w_{ij} > 0, i\in S\}$.

\begin{assumption}[Sufficiency Condition]\label{a:suff}
For any subset $S$ of agents, if $\Gamma(S) \neq G$ then there exists $i\in S$ and $j \in G\setminus \Gamma(S)$ such that $u_{ij1} > 0$. 
\end{assumption}

Let us first characterize the demand of each agent $i$ under spending constraint utilities. Given nonzero prices $\vecp$, define the bang-per-buck relative to $\vecp$ for segment $k$ in $f_{ij}$ to be $u_{ijk}/p_j$. Sort all segments of agent $i$ by decreasing bang-per-buck value, and partition them by equality into classes: $Q_1, Q_2, \ldots, Q_t$. Then an allocation is a demand bundle of agent $i$ if and only if there is an integer $t_i$ such that all segments in partitions $Q_1, \ldots, Q_{t_i-1}$ are all fully allocated, and no segments in partitions $Q_{t_i+1}, Q_{t_i+2}, \ldots$ are allocated. Furthermore, the total money spent on partitions $Q_1, \ldots, Q_{t_i-1}$ is no more than agent $i$'s total budget $m_i = \sum_{j \in G}p_jw_{ij}$. We term $Q_{t_i}$ agent $i$'s {\it current partition}, and $Q_1, \ldots, Q_{t_i-1}$ agent $i$'s {\it allocated partition}. Let $spent^a_i$ denote the total money of agent $i$ spent on allocated segments, and let $spent^g_j$ denote the total money spent on allocated segments of good $j$. Agent $i$ can freely demand any segments in her current partition, fully or partially, until her remaining budget $m_i - spent^a_i$ is exhausted.

\paragraph{Equality and Ratio Graph} The main tool for the analysis of previous algorithms in spending constraint markets is an {\it equality graph}, denoted by $EG(\vecp)$. This graph remains completely unknown to our algorithm, but it proves useful when proving properties of the convergence process. The vertex set of this bipartite graph consists of the set of agents $A$ and the set of goods $G$. Given a price vector $\vecp_t$, we introduce an edge from agent $i$ to good $j$ if and only if agent $i$'s current partition $Q_{t_i}$ contains one segment that belongs to utility function $f_{ij}$ for good $j$. The edges in $EG(\vecp)$ are called \emph{equality edges}.
Observe that this graph changes throughout the process when we update the price vector $\vecp_t$. 

Based on $EG(\vecp)$ we can construct an {\it equality network} denoted by $N(\vecp)$: First for each edge in $EG(\vecp)$ from agent $i$ to good $j$, let $k$ be the corresponding segment for $f_{ij}$ that belongs to $Q_{t_i}$. We assign a capacity of $c_{ij} = B_{ijk}m_i$ to this edge. Next add a source vertex $s$ and a sink vertex $t$. For each agent $i$, add an edge from $s$ to $i$ with capacity $m_i-spent^a_i$, and finally add an edge from every good to sink $t$ with infinite capacity. It is easy to see that every maximum flow in $N(\vecp)$ corresponds to a feasible demand allocation for each agent.
 
Similar to~\cite{DevanurPSV08}, we define a {\it balanced flow} as a maximum flow in $N(\vecp)$ that minimizes $\sum_{j \in G}(\ell_{jt} + spent^g_j - p_j)^2$, where $\ell_{jt}$ is the flow along edge $(j, t)$, which also denotes the money spent on good $j$ on segments of current partitions, and $spent^g_j$ is the amount of money spent on allocated partitions of good $j$. By assumption, $\Oracle(\vecp)$ returns the surplus vector derived from (any) balanced flow of the network $N(\vecp)$\footnote{It can be easily shown that every balanced flow gives the same surplus vector.}.

As mentioned above, our algorithm cannot see $EG(\vecp)$. We use a different structure that can be observed on the goods side.

\begin{definition} 
The {\it ratio graph} $RG(M, \vecp)$ is an undirected graph with $m$ vertices (where $m$ is the number of goods in the market), and for any two goods $i$ and $j$, $(i, j)$ is an edge if and only if $p_i/p_j$ can be represented as a ratio of two integers, each of value at most $2^M - 1$.
\end{definition}

We can compute the ratio graph using only the price vector and the input size bound. It allows us to retrieve some information about the hidden structure of $EG(\vecp)$. 


\begin{claim} \label{claim:rg}
  Let $L$ be the upper bound on the number of bits to represent each utility parameter and $\vecp$ be a price vector. 
  For any price vector $\vecp$ and goods $i, j$ that are connected in $EG(\vecp)$, $i$ and $j$ are also connected in $RG(M, \vecp)$ for any $M \geq L$.
\end{claim}
\begin{proof}
If good $i$ and good $j$ are connected in $EG(\vecp)$, then there exist goods $i = i_0, i_1, \ldots, i_{k-1}, i_k = j$, such that for each $t < k$, there exists some agent $a_t$ that has the same bang-per-buck for two segments for goods $i_t$ and $i_{t+1}$, respectively. Then we have $p_{i_t}/p_{i_{t+1}} = u_{a_ti_tk_1}/u_{a_ti_{t+1}k_2}$ for some $k_1$ and $k_2$. This is a ratio of two integers, each of value at most $2^L-1$. This implies $(i_t, i_{t+1}) \in RG(M, \vecp)$, so $i$ and $j$ are connected in $RG(M, \vecp)$.
\end{proof}
There exists an algorithm for computing exact equilibrium prices in Fisher markets with spending constraint utilities~\cite{Vazirani10}.~\cite{DevanurV04} give an FPTAS for exchange markets with spending constraint utilities. The algorithm finds an $(1+\varepsilon)$-approximate equilibrium in time polynomial in the input size and $1/\varepsilon$.



For spending constraint markets, we extend our approach for WGS markets. The challenge is that surplus can change in a non-continuous way when prices change the current partitions of the agents. However, we show how to use the linear structure of the market to get rough information about these breakpoints. Also, we maintain prices within a polynomial precision and guarantee convergence to an approximate market equilibrium. Finally, when $\Phi(\vecp)$ becomes small enough, we convert the approximate equilibrium to an exact one using a procedure \AlgSCExact.

Our only assumption is that the whole input can be represented within $L$ bits, and $L$ is known to the algorithm. This implies as a corollary a variant of Assumption~\ref{asp1} -- there is an exact market equilibrium with prices $\vecp$ s.t. $\frac{\max_ip_i}{\min_ip_i} \leq 2^{D_1}$ where $D_1$ is a polynomial of $m$ and $L$. As mentioned before, Assumption~\ref{asp2} does not hold in spending constraint markets.


\subsection{The Framework}

\begin{algorithm*}[t]
\caption{\AlgSC: \label{alg:NoPrecisionLinear} Framework for Spending Constraint Markets}
\DontPrintSemicolon
\SetKwInOut{Input}{Input}\SetKwInOut{Output}{Output}
\SetKw{Param}{Parameters:}
\Input{number of goods $m$, demand oracle $\Oracle$, number of bits $L$ to represent the whole input (including each $u_{ijk}, B_{ijk}, w_{ij}$); solution precision $\varepsilon$}
\Output{Prices $\vecp$ of a $(1+\varepsilon)$-approximate market equilibrium}
\nonl \Param{$\varepsilon' = \frac{\varepsilon}{2\sqrt{m}}$}
\BlankLine
Set initial price $\vecp_0 \ot (1, 1, \ldots, 1)$ and round index $t \ot 0$.\;
\Repeat( \ //round $t$){$\|\Oracle(\vecp_t)\|^2 < \varepsilon'^2$}{
$t \ot t+1$\;
$\vecs = (s_1, \ldots, s_m) \ot \Oracle(\vecp_{t-1})$\;
Sort $\vecs$ such that $s_{i_1} \geq s_{i_2} \geq \cdots \geq s_{i_m}$.\;
Find smallest $k$, such that $s_{i_{k+1}} \leq 0$ or $s_{i_k} >
(1+\frac1m)s_{i_{k+1}}$. \;
Set $G_1 \ot \{i_1, \ldots, i_k\}$ and $G_2 \ot [m] \backslash G_1$\;
Binary search the smallest $x \in (1, \infty)$, such that in $\vecs' = \Oracle(\UpdatePrice(\vecp_{t-1}, x, G_1))$ it holds \\
\nonl \hspace{0.5cm} $\min\{s'_i \mid i \in G_1\} \le \max\{\{s'_i \mid i \in G_2\} \cup \{0\}\}$.\;
$\vecp_{t} \ot \UpdatePrice(\vecp_{t-1}, x, G_1)$\;
}
\Return{$\vecp_t$}
\end{algorithm*}

\AlgSC specifies the general framework of our algorithm, which is similar to the approach taken previously for WGS markets. As mentioned in Section~\ref{sec:model}, here we do not resort to approximation parameter $\mu$, but instead compute the exact surplus. We first analyze \AlgSC and show that it needs only a polynomial number of rounds to converge to an approximate market equilibrium. For now, our analysis disregards all precision and representation issues. In particular, we assume to find the exact value $x$ using binary search, irrespective of the number of bits needed for representation. Also, the update of prices from $\vecp_{t-1}$ to $\vecp_t$ will multiply all prices of goods in $G_1$ by $x$, irrespective of the number of bits required to represent them.
In our final algorithm below, we will show how to address these issues to obtain a (true) polynomial-time algorithm. 

The analysis of \AlgSC proceeds roughly as in the previous section. We rely on the following lemma.

\begin{lemma}
	\label{lem:scIsWGS}
	The demands returned by $\Oracle(\vecp)$ satisfy the WGS property.
\end{lemma}

\begin{proof}
We make use of a max-min fair property for balanced flows in linear markets proved in previous work. A vector $\vecs$ is called \emph{max-min fair} iff for every feasible vector $\vecs'$ and $i$ such that $s_i < s'_i$, there is some $j$ with $s_j < s_i$, such that $s_j > s'_j$. The following claim is proved by~\cite{DevanurPSV08}.
\begin{claim}~\cite{DevanurPSV08}
  The surplus vector of a balanced flow in $N(G)$ is max-min fair among all feasible surplus vectors. 
\end{claim}
Although this claim is for linear markets, it can also be directly applied to spending constraint markets because the network flows are designed only with respect to the current partition of each agent $i$. Within this domain the spending constraint market behaves exactly like a linear market.

We proceed to prove Lemma~\ref{lem:scIsWGS} by contradiction to this claim. Suppose we increase the price of some good $k$ from $p_k$ to $p_k + \delta$. We denote the old prices by $\vecp$ and the new prices by $\vecp'$. Let $\vecs = \Oracle(\vecp)$ and $\vecs' = \Oracle(\vecp')$. Now assume for contradiction that there exists some $\ell \neq k$ such that $s'_\ell < s_\ell$. Let $S_G = \{j \in G \mid s_j \geq s_\ell, j\neq k\}$ and $S_A = \{i \in A \mid \text{there exists } j\in S_G \text{ such that }f_{ij} > 0 \text{ in }N(\vecp)\}$. It is easy to verify the following properties. 
\begin{itemize}
%
\item [(a)] For any edge $(i,j)$ in $N(\vecp)$ with $j \neq k$, let $seg$ be the segment in agent $i$'s current partition that corresponds to this edge. Then with new prices $\vecp'$, either $seg$ belongs to the allocated partition (i.e., is fully allocated), or $seg$ is still in agent $i$'s current partition (i.e., $(i,j)$ also exists in $N(\vecp')$).
\item [(b)] For any edge $(i, j) \in N(\vecp')$ with $i \in S_A$ and $j \not\in S_G$, let $seg$ be the segment in agent $i$'s current partition that corresponds to this edge with prices $\vecp'$. If $seg$ is also in $i$'s current partition with prices $\vecp$, then $seg$ is fully allocated with price $\vecp$ (otherwise agent $i$ can reroute some flow from $S_G$ to this segment to obtain a more balanced flow in $N(\vecp)$). If $seg$ is not in $i$'s current partition with price $\vecp$, then all segments in $i$'s current partition with price $\vecp$ must be fully allocated in the demand allocation with prices $\vecp'$.
%
%
\end{itemize}
The above two observations imply that with prices $\vecp'$, we can rearrange flows from $S_A$ to $S_G$ to get a new feasible surplus vector $\vecs''$ such that for all $j \in S_G$, $s''_j \geq s_j$, and for all other edges $j \notin S_G$, $s''_j = s'_j$. 
In particular, we have $s''_j \geq s_j \geq s_\ell > s'_\ell$ for all $j \in S_G$. This contradicts the fact that $\vecs'$ is max-min fair and proves Lemma~\ref{lem:scIsWGS}.
\end{proof}
Now the following properties can be proved using literally the same proofs as for WGS markets before.
\begin{claim} \label{claim:linear:summary}
	For \AlgSC the following properties hold.
\begin{enumerate}
	\item	In \AlgSC, $|\Oracle(\vecp_0)| \leq 2m$, and $|\Oracle(\vecp_t)|$ is non-increasing in $t$.
  \item For any price vector $\vecp$, $x > 1$, and $S \subseteq [m]$ it holds $\Oracle(\UpdatePrice(\vecp, x, S))_i \leq x \cdot \Oracle(\vecp)_i$ for any $i \in S$.
	\item The number of rounds that end with $x \geq 1+\frac1{Rm^3}$ in \AlgSC is $O(m^4D_1)$, for a sufficiently large constant $R$.
\end{enumerate}
\end{claim}
We also show a version of Claim~\ref{claim:price1} for spending constraint markets, which needs some extra work. Unlike for WGS markets, the surplus of a good can change from non-negative to negative. Thus, the proof of Claim~\ref{claim:price1} does not directly transfer to spending constraint markets. Instead, we first show the following.
\begin{claim} \label{claim:negsubset}
  Let $S_t$ be the set of goods with negative surplus and price strictly greater than 1 at the end of round $t$ in \AlgSC. For any $T \subseteq S_t$, let $\Gamma(T, \vecp_t)$ be the neighbors of set $T$ in $EG(\vecp_t)$, i.e., $\Gamma(T, \vecp_t)$ is the set of agents who are interested in at least one good in $T$ under price $\vecp_t$. Let $B(\Gamma(T, \vecp_t))$ be the sum of budgets of agents in $\Gamma(T, \vecp_t)$. Then we have $B(\Gamma(T, \vecp_t)) > \sum_{i \in T}(\vecp_t)_i$.
\end{claim}
\begin{proof}
  We prove this claim by induction. The claim is trivially true for round 0. Assume that it is true for any round $t \leq t'$, then at the end of round $t = t'+1$, consider two cases:
  \begin{itemize}
  \item $\min\{s'_i \mid i \in G_1\} \geq 0$. Because the surplus of any good in $G_2$ is non-decreasing in round $t$, we have $S_t \subseteq S_{t-1}$. Further, the algorithm does not increase the price of any good in $T$ in round $t$. Hence, we also have $\Gamma(T, \vecp_{t-1}) \subseteq \Gamma(T, \vecp_t)$. By the induction assumption, we have $B(\Gamma(T, \vecp_t)) \geq B(\Gamma(T, \vecp_{t-1})) > \sum_{i \in T}(\vecp_{t-1})_i = \sum_{i \in T}(\vecp_t)_i$.
  \item $\min\{s'_i \mid i \in G_1\} < 0$. 
  Assume that during this round we start from $x=1$ and increase $x$ continuously until it reaches its final value. We also assume that the equality graph $EG(\UpdatePrice(\vecp_{t-1}, x, G_1))$ and the corresponding balanced flow are implicitly being maintained throughout the process. For each agent $i$ and any moment during this round, let $\Gamma(i)$ denote the neighbors of agent $i$ in the equality graph. There are two cases:

\paragraph{Case (a):} $\Gamma(i) \cap G_1 \neq \emptyset$ and $\Gamma(i) \cap G_2 \neq \emptyset$. This means agent $i$'s current partition contains segments of goods in both $G_1$ and $G_2$. When we continue to increase $x$ (and consequently the prices of goods in $G_1$), the segments of goods in $G_1$ will have worse bang-per-buck value than segments of goods in $G_2$. Hence, they will be removed from agent $i$'s current partition. Furthermore, for any good $j_1 \in \Gamma(i) \cap G_1$ and $j_2 \in \Gamma(i) \cap G_2$, because we have $s_{j_1} > s_{j_2}$ and the balanced flow condition, it cannot happen that $\ell_{ij_1} > 0$ and $\ell_{ij_2} < c_{ij_2}$. Otherwise agent $i$ would be able to route some flow from edge $(i, j_1)$ to edge $(i, j_2)$ to reach a more balanced flow. Thus, two possibilities remain:
  	\begin{itemize}
    \item $\ell_{ij_2} = c_{ij_2}$ for every $j_2 \in \Gamma(i) \cap G_2$. This means all segments of goods in $G_2$ are already fully allocated. When $x$ increases, the new current partition of agent $i$ will only contain segments of goods in $\Gamma(i) \cap G_1$, and the surpluses will change continuously with the change of $x$.
    \item $\ell_{ij_1} = 0$ for every $j_1 \in \Gamma(i) \cap G_1$. This means when $x$ increases, the segments that are being removed from $i$'s current partition are all unallocated in the current allocation. Hence, the surpluses will again change continuously at the current point.
    \end{itemize}
  In summary, in this case the change of the surpluses will always be continuous in the change of $x$. Note that when the surpluses are changing continuously, Case 2 cannot occur before Case 1 happens. Therefore, the Claim follows via Case 1.
          
\paragraph{Case (b):} $\Gamma(i) \subseteq G_1$ or $\Gamma(i) \subseteq G_2$. This means we will always change the prices of all segments in $i$'s current partition by the same rate. Hence, if set $Q_{t_i}$ changes, it can only be merged with some segments in $Q_{t_i-1}$ or $Q_{t_i+1}$. In either case, the final value of $x$ must be at a point where at least one new edge emerges in the equality graph $EG(\vecp_t)$ from $\Gamma(G_1) \subseteq A$ (the set of agents incident to at least one good in $G_1$) to $G_2$. 
Suppose we alter the equality graph by removing these emerging edges from $EG(\vecp_t)$. If we recompute the balanced flow, then $\min\{s'_i \mid i \in G_1\} > \max\{\{s'_i \mid i \in G_2\} \cup \{0\}\}$. For any $T \subseteq S_t$, let $T_1 = T \cap G_1$ and $T_2 = T\cap G_2$. 
In this new graph, let $\Gamma'(T_1)$ be the set of agents who have positive flow to at least one good in $T_1$, and let $\Gamma'(T_2)$ be the set of
agents incident to at least one good in $T_2$. Since $\min\{s'_i \mid i \in T_1\} > \max\{s'_i \mid i \in T_2\}$, by the balanced flow condition we know $\Gamma'(T_1) \cap \Gamma'(T_2) = \emptyset$. Also, we have $B(\Gamma'(T_1)) > \sum_{i \in T_1}(\vecp_t)_i$ because every good in $T_1$ has positive surplus, and $B(\Gamma'(T_2)) > \sum_{i \in T_2}(\vecp_t)_i$ because the claim is true in round $t-1$. Combining these two inequalities gives us $B(\Gamma(T, \vecp_t)) > \sum_{i \in T}(\vecp_t)_i$.
  \end{itemize}
\end{proof}
\begin{corollary}\label{cor:connected}
  At the end of each round $t$ in \AlgSC, for any good $i$ with negative surplus and price greater than 1, there exists another good $j$ with price 1 that is connected to $i$ in $EG(\vecp_t)$
\end{corollary}
\begin{proof}
  Assume for contradiction that the statement is false. Let $T$ be the set of goods with negative surplus and connected with good $i$ in $EG(\vecp_t)$. By the balanced flow condition, none of the agents in $\Gamma(T, \vecp_t)$ can have any positive flow to goods outside set $T$. Thus we have $0 > \sum_{i \in T}s'_i = B(\Gamma(T, \vecp_t)) - \sum_{i \in T}(\vecp_t)_i$. This contradicts Claim~\ref{claim:negsubset}.
\end{proof}
We obtain the following corollary, an analog of Claim~\ref{claim:price1} for spending constraint markets.
\begin{corollary}\label{cor:price1}
  Throughout the run of \AlgSC, there will be at least one good whose price remains 1.
\end{corollary}
The set of properties shown so far allows to establish the following lemma, which is the key step for observing convergence to equilibrium. It can be seen as an adjustment of Lemma~\ref{lemma:potential} to spending constraint markets. As before, it involves a sufficiently large constant $R$. 
\begin{lemma} \label{lemma:linear:potential}
 If $x < (1+\frac{1}{Rm^3})$ at the end of round $t$ in \AlgSC, then $\Phi(\vecp_t) \leq \Phi(\vecp_{t-1}) \left(1-\Omega\left(\frac{1}{m^3}\right)\right)$.
\end{lemma}
\begin{proof}
  Our proof uses the arguments of the proof for Lemma~\ref{lemma:potential}. We consider two cases:

	\paragraph{Case 1: $\min\{s'_i \mid i \in G_1\} = \max\{\{s'_i \mid i \in G_2\} \cup \{0\}\}$} 
	This case can be verified by observing that Claim~\ref{append:claim:minG1} holds with $\mu = 0$. Then the proof follows using exactly the same proof as for Lemma~\ref{lemma:potential}.
	
	\paragraph{Case 2: $\min\{s'_i \mid i \in G_1\} < \max\{\{s'_i \mid i \in G_2\} \cup \{0\}\}$}
        Using the same argument as in the proof of Claim~\ref{claim:negsubset}, one can show that in this case $x$ must be at a point where at least one new edge emerges in the equality graph $EG(\vecp_t)$ from $\Gamma(G_1) \subseteq A$. Without these emerging edges, we have $\min\{s'_i \mid i \in G_1\} \ge \max\{s'_i \mid i \in G_2\}$. This means that we can further reduce the flows along the new edges to get another feasible flow in $N(\vecp_t)$, such that with the resulting surplus vector $s''$ we have $\min\{s''_i \mid i \in G_1\} = \max\{s''_i \mid i \in G_2\}$. 

        Next, using again the same proof as for Lemma~\ref{lemma:potential}, for a sufficiently large constant $R$ it follows that $\|s''\|^2 \leq \Phi(\vecp_{t-1})\left(1-\Omega\left(\frac{1}{m^3}\right)\right)$. Hence, $\Phi(\vecp_t) \leq \|s''\|^2 \leq \Phi(\vecp_{t-1})\left(1-\Omega\left(\frac{1}{m^3}\right)\right)$.
\end{proof}
\subsection{Precision and Representation} \label{subsec:precision}
It is now tempting to think that using a similar argument as in the general WGS case, Claim~\ref{claim:linear:summary}, Corollary~\ref{cor:price1} and Lemma~\ref{lemma:linear:potential} provide an algorithm that converges to an approximate market equilibrium in polynomial time. However, an issue arises with regard to the precision and representation of the prices: In each round, $x$ could potentially be a rational number involving prices and surpluses, and after multiplying each price in $G_1$ by $x$, the bit length to represent a price can double in one round. This means that after a polynomial number of rounds, we may require an exponential number of bits to represent the prices of some goods as well as the desired factor $x$. 


%

Recall the ratio graph and Claim~\ref{claim:rg}. As an adjustment, we run the following procedure $\Rounding(\vecp, M)$ at the end of each iteration of the main loop in \AlgSC. Its purpose is to round the prices within polynomial bit length while maintaining the structure of equality and ratio graphs. Thereby, we will show that the value of potential function $\Phi(\vecp)$ will not be increased dramatically.

\begin{algorithm}[t]
\caption{\Rounding$(\vecp, M)$: Rounding procedure \label{alg:Rounding}}
\DontPrintSemicolon
\SetKwInOut{Input}{Input}\SetKwInOut{Output}{Output}
\SetKw{Param}{Parameters:}
\Input{price vector $\vecp$ in which $\min_i\vecp_i = 1$, rounding bound $M\geq L$}
\Output{Rounded price vector $\vecp'$}
\BlankLine
Let $\mathcal{P} = \{\frac{a}{b} \mid \textrm{$a, b \in \mathbb{Z}^+, a, b \leq 2^M$}\}$ and $\vecp' \ot \vecp$.\;
\While{$RG(M, \vecp')$ is not connected}{
  Let $C_1, C_2, \ldots, C_k$ be the connected components of $RG(M, \vecp')$\;
  Assume without loss of generality that $C_1$ is a component with $\min_{i\in C_1 \cap A} p'_i = 1$\;
  For every $i, j$, let $r_{ij} = p'_i / p'_j$. \;
  Let $b_{ij} = \min\{\frac{x}{r_{ij}} \mid x \in \mathcal{P}, x \geq r_{ij}\}$ and $B_{ij} = \min\{b_{i'j'} \mid i' \in C_i, j' \in C_j\}$.\;
  Let $(i, j) \in \arg\min_{i>j}B_{ij}$.\;
  $\vecp' \ot \UpdatePrice(\vecp', B_{ij}, C_i)$
}
\Return{$\vecp'$}
\end{algorithm}

Our new algorithm \AlgSCRounding is simply the framework \AlgSC with the following modifications: 
\begin{itemize}
\item [(1)] Set $M = \log_2{\frac{5m^7}{\varepsilon'^2}}$. In each round, binary search $x$ within domain $\mathcal{P} = \{\frac{a}{b} \mid a > b, a, b \in \mathbb{Z}^+, a, b \leq 2^{2mM+L}\}$, instead of $(1, \infty)$.
\item [(2)] At the end of each round $t$, update $\vecp_t \ot \Rounding(\vecp_t, M)$ with $M$ as above.
\end{itemize}
\begin{lemma}\label{lemma:rounding}
  Given any price vector $\vecp$, $\Rounding(\vecp, M)$ terminates in time polynomial in $m, M$. The returned price vector $\vecp'$ satisfies:
  \begin{itemize}
    \item[(a)] there exists $i \in G$ with $p_i = p'_i = 1$,
    \item[(b)] every price can be represented as a ratio of two integers, each of value at most $2^{mM}$,
    \item[(c)] $EG(\vecp')$ contains every edge present in $EG(\vecp)$, and
    \item[(d)] $p_i \leq p'_i \leq p_i + 2^{-M}$ for every $i \in G$.
  \end{itemize}
\end{lemma}
\begin{proof}
  Property (a) holds since we never increase the price of goods in set $C_1$, and property (b) can be derived based on property (a) and the fact that
  $RG(M, \vecp')$ is connected when the algorithm terminates.

  Next we claim that for any $i, j$, there does not exist any $x \in \mathcal{P}$, such that $p_i/p_j \leq x < p'_i/p'_j$. This is because by design, $p_i/p_j$ has to reach $x$ in some iteration before it grows beyond $x$. But starting from that moment until the end of the algorithm, $i$ and $j$ will be connected in $RG(M, \vecp)$. The ratios between two prices in the same connected component in $RG(M, \vecp)$ remain unchanged. Hence $p_i/p_j$ will never grow beyond $x$.

  This claim also proves property (c). By Claim~\ref{claim:rg}, a pair of goods $i$ and $j$ connected in $EG(\vecp)$ remain connected in $RG(M, \vecp)$. Hence, their ratio of prices will remain the same in $\vecp'$, and they are connected in $EG(\vecp')$. 

  For property (d), according to the algorithm no price will decrease from $\vecp$ to $\vecp'$. Next, number the goods such that $p_1 = p'_1 = 1$. For any $i \neq 1$, let $x_i = \min\{x \in \mathcal{P} \mid x \geq p_i\}$. Then we have $p'_i/p'_1 = p'_i \leq x_i$, as well as $x_i \leq p_i + 2^{-M}$. Therefore  $p_i \leq p'_i \leq p_i + 2^{-M}$.

  Finally, in each iteration we add at least one edge between two connected components in $RG(M, \vecp)$. Thus the algorithm will terminate after at most $m-1$ iterations, and it is easy to check that each iteration runs in polynomial time. This proves the claim.
\end{proof}
Once the prices are bounded by a fixed polynomial bit length, we can also bound the length needed to encode the desired $x$ in each round. This implies that we can find $x$ in the framework in polynomial time using binary search.
\begin{lemma}\label{lemma:roundingrational}
	In every round of \AlgSCRounding, the desired $x$ can be represented as a ratio of two integers, each of value at most $2^{mM+L+2\log{m}}-1$.
\end{lemma}
\begin{proof}
	Let $\vecp = \vecp_{t-1}$ and consider the structure of $EG(\vecp)$. We let $A_1 = \Gamma(G_1)$ be the agents connected to goods in $G_1$. Also, let $A_2 = A \setminus A_1$. There is no $(i,j) \in EG(\vecp)$ with $i \in A_1$ and $j \in G_2$, since otherwise $\Oracle(\vecp_{t-1})$ could increase the money spent on goods in $G_2$ and further decrease $\Phi(\vecp_{t-1})$. For simplicity, we will also assume that there is no edge $(i,j) \in A_2 \times G_1$, since no agent spends money along these edges and they immediately disappear once we start increasing prices in $G_1$.
	
	Now we increase $\vecp$ on goods in $j \in G_1$ by $x$ and get a new price vector $\vecp(x)$. This only generates new edges $(i,j) \in A_1 \times G_2$. Furthermore, we drop only edges $(i,j) \in A_1 \times G_1$. To verify this let us consider the other possibilities. The relation between marginal utility values $u_{ijk}/p_j$ and $u_{ij'k'}/p_{j'}$ for goods in $j,j' \in G_1$ does not change, since both $p_j$ and $p_k$ are both multiplied by $x$. Hence, there are no new edges $(i,j) \in A_1 \times G_1$. For the same reason, there are no new edges $(i,j) \in A_2 \times G_2$. The bang-per-buck of goods in $G_1$ decreases, so we also do not introduce edges $(i,j) \in A_2 \times G_1$. In fact, this also implies that we do not drop any  edges $(i,j) \in A_2 \times G_2$ -- prices and bang-per-buck relations among goods in $G_2$ do not change at all and $G_2$ becomes more attractive compared to $G_1$. Finally, there exist no edges $(i,j) \in (A_1 \times G_2) \cup (A_2 \times G_1)$ that could be removed. This shows that we only generate new edges between $A_1$ and $G_2$, and we only drop edges between $A_1$ and $G_1$.
	
	For any given $x$, consider the residual graph of $N(\UpdatePrice(\vecp, x, G_1))$. Let $C$ be an arbitrary connected component in this graph. Let $C_A$ be the set of agents in $C$ and $C_G$ be the set of goods in $C$. Then we know that all goods in $C_G$ have the same surplus, and all flow going through $C_G$ comes from agents in $C_A$. This implies the following equation:
	$$\sum_{i \in C_A \backslash G_1} p_i + x \sum_{i \in C_A \cap G_1}p_i = \sum_{i \in C_G \backslash G_1}p_i + x\sum_{i \in C_G \cap G_1}p_i + |C_G|s\enspace,$$
where $s$ is the surplus of (any) good in this component. 

  Now let us focus on the moment where $x$ reaches the desired value at the end of round $t$ according to the algorithm. At this moment, one of the following properties must hold:
	\begin{itemize}
  \item[(1)] $\min\{s'_i \mid i\in G_1\} = 0$. Then for the connected component that contains a good of surplus 0, we have $s = 0$ in the above equation. All initial prices $p_i$ are ratios of integers with values at most $2^{mM}$, so when we solve the equation for $x$, the solution is a ratio of two integers with value at most $m2^{mM} = 2^{mM+\log{m}}$.
  \item[(2)] $\min\{s'_i \mid i\in G_1\} = \max\{s'_j \mid j \in G_2\}$. In this case, we have two possibilities:
  	\begin{itemize}
    \item There exist two connected components in the residual graph of $N(\UpdatePrice(\vecp, x, G_1))$ that have the same surplus. Applying the above equation to these two components, we can solve for $x$, and the solution will be a ratio of two integers with value at most $m^22^{mM} = 2^{mM+2\log{m}}$.
    \item A new edge $(i,j) \in A_1 \times G_2$ appears, then for agent $i$ good $j \in G_2$ becomes equally attractive as some $k \in G_1$: $u_{ij}/p_{j}x = u_{ik}/p_{k}$, or, equivalently, $x = (u_{ij}p_{k})/(u_{ik}p_{j})$. $p_{ij}$ and $p_{k}$ can be represented as ratio of integers of value at most $2^{mM}$ by Lemma~\ref{lemma:rounding}, and $u_{ij}$ and $u_{ik}$ are both integers of value at most $2^L-1$. Hence, every value of $x$ at which a new edge evolves in $EG(\vecp(x))$ can be presented as a ratio of integers of value at most $2^{2mM+L}$.
    \end{itemize}
	\end{itemize}
\end{proof}
We now bound the impact of replacing $\vecp_t$ by $\Rounding(\vecp_t, M)$ in the function $\Phi(\vecp_t)$.
\begin{lemma} \label{lemma:roundingbound}
	$\Phi(\Rounding(\vecp_t, M)) < \Phi(\vecp_t) + 5m^32^{-M}$ for any round $t$.
\end{lemma}
\begin{proof}
  Let $\vecp'_t = \Rounding(\vecp_t, M)$. By Lemma~\ref{lemma:rounding}(c), we know this rounding procedure does not remove any edges in $EG(\vecp_t)$. Let $f$ be a balanced flow in $N(\vecp_t)$. Then by Lemma~\ref{lemma:rounding}(d) we can construct a feasible flow $f'$ in $N(\Rounding(\vecp_t, M))$, such that $f_{ij} \leq f'_{ij} \leq f_{ij} + 2^{-M}$ for every $i, j$. Let $\vecs$ be the surplus vector derived from $f'$, then we have $|s_i - \Oracle(\vecp_t)_i| < m2^{-M}$ for every $i$.
Hence 
\begin{eqnarray*}
 \|\vecs\|_2^2 - \|\Oracle(\vecp_t)\|_2^2 & = & \sum_i(s_i^2 - \Oracle(\vecp_t)^2_i) \\
                                & \leq & \sum_i(2m2^{-M}|\Oracle(\vecp_t)_i| + m^22^{-2M}) \\
                                & = & 2m^22^{-M}|\Oracle(\vecp_t)| + m^32^{-2M} \\
                                & \leq & 4m^32^{-M} + m^32^{-2M} \\
                                & < & 5m^32^{-M}\enspace.
\end{eqnarray*}

Note that $\vecs$ is just one feasible surplus vector for price vector $\vecp'_t$, and $\Phi(\vecp'_t)$ minimizes the $\ell_2$-norm of surpluses among all feasible surplus vectors. Hence, $\Phi(\Rounding(\vecp_t, M)) \le \|\vecs\|_2^2 < \|\Oracle(\vecp_t)\|_2^2 + 5m^32^{-M}$. This proves the lemma.
\end{proof}
\Rounding can be used to obtain an algorithm that converges to an approximate market equilibrium in polynomial time.  
\begin{lemma}
	\label{lem:apxlinear}
  For any spending constraint exchange market satisfying Assumption \ref{a:suff}, an $(1+\varepsilon)$-approximate market equilibrium can be computed in time polynomial in $m$, $L$, and $\log(1/\varepsilon)$.
\end{lemma}
\begin{proof}
In the \AlgSC framework with \AlgSCRounding, we know by Lemma~\ref{lemma:linear:potential} that at the end of each round $t$ and before calling \Rounding, $\Phi(\vecp_t) \leq \Phi(\vecp_{t-1})\left(1-\Omega\left(\frac{1}{m^3}\right)\right)$. If $\Phi(\vecp_t) > \varepsilon'^2$, we have $5m^32^{-M} = \varepsilon'^2/m^4 < \Phi(\vecp_t)/m^4$. Thus by Lemma~\ref{lemma:roundingbound},
\begin{eqnarray*}
  \Phi(\Rounding(\vecp_t, M)) & \leq & \left(1 + \frac{1}{m^4}\right)\Phi(\vecp_t) \\
                              & \leq & \left(1 + \frac{1}{m^4}\right) \left(1 - \Omega\left(\frac{1}{m^3}\right)\right)\Phi(\vecp_{t-1}) \\
                              & = & \left(1 - \Omega\left(\frac{1}{m^3}\right)\right)\Phi(\vecp_{t-1})
\end{eqnarray*}
This implies we can employ the same proof as for Theorem~\ref{thm:exchangeWGS} to show that after finishing \AlgSCRounding, we arrive at a $(1+\varepsilon)$-approximate market equilibrium. Because $M$ is a polynomial in the input size and $\log(1/\varepsilon)$, the binary search and the \Rounding procedure run in polynomial time in each round of the framework. Hence, the running time is polynomial in the input size and $\log(1/\varepsilon)$.
\end{proof}

Finally, it remains to convert the approximate equilibrium to an exact one. To achieve this, we rely on full information about the spending constraint utilities. While this step can be seen as an extension of the technique developed in~\cite{DuanM15} for the linear exchange markets, there are several challenges due to the much more involved setting of spending constraint utilities, where the allocated partitions make the remaining budgets of agents and the values of goods dependent on too many parameters. Using a more involved procedure we are able to handle the extra complexity of the problem. Our result resolves an open question of~\cite{DuanM15} of finding an exact polynomial time algorithm for exchange markets with spending constraint utilities. A detailed discussion of this final step can be found in Appendix~\ref{app:rounding}. This yields the final theorem in this section.

\begin{theorem}
  \label{thm:exchangelinear}
	For any spending constraint exchange market satisfying Assumption~\ref{a:suff}, \AlgSCExact returns the price vector of a market equilibrium in time polynomial in $m$ and $L$.
\end{theorem}

\bibliographystyle{plain}
\bibliography{../../../../Bibfiles/literature}

\clearpage
\appendix


\section{WGS Exchange Markets}
\label{app:wgs}

In this section we describe the complete algorithm \AlgExchange\ for WGS exchange markets. Recall the assumptions from Section~\ref{sec:wgs}, which we restate here for completeness. 


\medskip

{\noindent \bf Assumption~\ref{asp1}}
{\it There exists a market equilibrium $(\vecp^*,\vecx^*)$ with $1 \le p_i^* \le 2^{D_1}, \forall i \in G$.}

\medskip

{\noindent \bf Assumption~\ref{asp2}}
{\it For any price vector $\vecp$ such that $1 \leq p_i \leq 2^{D_1}$ for each $i$, $|\frac{\partial s_i}{\partial p_j}| < 2^{D_2}$ for every $i, j$, where $s_i$ is the surplus money of good $i$ in $\Oracle(\vecp)$, and $D_2$ is a polynomial of the input size.}

\medskip


\begin{algorithm}[t]
\caption{\AlgExchange\label{append:alg:WGS}}
\DontPrintSemicolon
\SetKwInOut{Input}{Input}\SetKwInOut{Output}{Output}
\SetKw{Param}{Parameters:}
\Input{number of goods $m$, approximate demand oracle $\AOracle$, precision bound $\varepsilon > 0$}
\Output{Prices $\vecp$ of a $(1+\varepsilon)$-approximate market equilibrium}
\nonl \Param{$\mu = \frac{\varepsilon}{R_1m^7}$, $\Delta = \frac{1}{\mu}(2^{D_1+D_2+\log{m}})$, $\varepsilon' = \frac{\varepsilon}{2\sqrt{m}}$ }

\BlankLine
Set initial price $\vecp_0 \ot (1, 1, \ldots, 1)$ and round index $t \ot 0$.\;
Let $\mathcal{P} = \{\frac{a}{b} \mid a > b, a, b \in \mathbb{Z}^+, a, b \leq \Delta\}$\;
\Repeat( \ //round $t$){$\|\AOracle(\vecp_t)\|^2 < \varepsilon'^2$}{
$t \ot t+1$\;
$\tvecs = (\ts_1, \ldots, \ts_m) \ot \AOracle(\vecp_{t-1}, \mu)$\;
Sort $\tvecs$ such that $\ts_{i_1} \geq \ts_{i_2} \geq \cdots \geq \ts_{i_m}$.\;
Find smallest $k$, such that $\ts_{i_{k+1}} \leq \mu$ or $\ts_{i_k} >
(1+\frac1m)\ts_{i_{k+1}}$. \;
Set $G_1 \ot \{i_1, \ldots, i_k\}$ and $G_2 \ot G \setminus G_1$\;
Binary search the largest $x \in \mathcal{P}$, such that in $\tvecs' = \AOracle(\UpdatePrice(\vecp_{t-1}, x, G_1))$ it holds\\ \nonl \hspace{0.5cm} $\min\{\ts'_i \mid i \in G_1\} \geq \max\{\{\ts'_i \mid i \in G_2\}\cup \{\mu\}\}$.\;
$\vecp_{t} \ot \UpdatePrice(\vecp_{t-1}, x, G_1)$\;
}
\Return{$\vecp_t$}
\end{algorithm}

Throughout the analysis and proofs below, if $\vecs = \Oracle(\vecp)$ for some $\vecp$, we use $\tilde{\vecs}$ to denote the surplus vector returned by the $\mu$-approximation demand oracle with the same price vector, i.e., $\tilde{\vecs} = \AOracle(\vecp, \mu)$. We proceed along similar lines as in Section~\ref{sec:wgs}, and the proofs of the first claims closely resemble the versions for the exact oracle. For completeness, we provide them here for the approximate oracle.

\begin{claim}\label{append:claim:surplusBound}
  In \AlgExchange, $|\Oracle(\vecp_0)| \leq 2m$ and $|\Oracle(\vecp_t)|$ is non-increasing in $t$.
\end{claim}

\begin{proof}
  Let $d_i$ be the exact demand for good $i$ under price $\vecp_0$, then $|\Oracle(\vecp_0)| = \sum_i|d_i - 1| \leq \sum_id_i + m = 2m$. Next, by the criteria to define $G_1$ and $G_2$ in each round, we have $\{i \mid \Oracle(\vecp_{t-1})_i < 0\} \subseteq G_2$: To see this, observe that the surplus resulting from the approximate $\AOracle(\vecp_{t-1})$ differs by at most an additive $\mu = \varepsilon/(R_1m^7)$, so a good $i$ with $\Oracle(\vecp_{t-1})_i < 0$ will always be classified in $G_2$ with respect to $\AOracle(\vecp_{t-1})_i$.
  
  During round $t$, only prices of goods in $G_1$ are increased. By the WGS property, we know $\Oracle(\vecp_{t})_i \geq \Oracle(\vecp_{t-1})_i$ for every $i \in G_2$. Further, note that $\min\{\Oracle(\vecp_t)_i \mid i \in G_1\} \geq 0$ since $\min\{\AOracle(\vecp_t)_i \mid i \in G_1\} \geq \mu$. Hence, we do not introduce any new negative surplus in $\Oracle(\vecp_{t})$. Thus, we have 
{\small
$$|\Oracle(\vecp_{t-1})| = -2\sum_{\Oracle(\vecp_{t-1})_i < 0}\Oracle(\vecp_{t-1})_i \geq -2\sum_{\Oracle(\vecp_{t})_i < 0}\Oracle(\vecp_{t-1})_i \geq
-2\sum_{\Oracle(\vecp_{t})_i < 0}\Oracle(\vecp_{t})_i = |\Oracle(\vecp_{t})|\enspace.$$}
\end{proof}

The next two claims bound the range of prices we encounter, which is important for showing that we approach the unique market equilibrium. 

\begin{claim}\label{append:claim:price1}
  Throughout the run of \AlgExchange, every good with negative surplus has price 1. Hence, there will be at least one good whose price remains 1.
\end{claim}
\begin{proof}
  Observe the following three simple facts about the surplus $\Oracle(\vecp_t)$ resulting from exact demands: (1) Throughout the algorithm we never increase the price of any good with negative surplus. (2) The surplus of any good does not change from non-negative to negative. (3) For any non-equilibrium price vector, there will always be a good with negative surplus. These facts are direct consequences of the conditions used to classify goods based on $\AOracle(\vecp_t)$ in the algorithm. Together they prove the claim.
\end{proof}

\begin{claim}\label{append:claim:priceUpper}
  In \AlgExchange, for any $t \geq 0$, all prices in $\vecp_t$ are bounded by $2^{D_1}$.
\end{claim}
\begin{proof}
  Let $\vecp^*$ be equilibrium prices according to Assumption~\ref{asp1}. We show that for any $t \geq 0$, $\vecp_t$ is always pointwise smaller than $\vecp^*$. Assume that this is not true, let $t$ be the smallest value such that there exists $(\vecp_{t})_i > \vecp^*_i$ for some $i$. Note that according to the algorithm, we have $\vecp_t = \UpdatePrice(\vecp_{t-1}, x, G_1)$ for some $x > 1$ and $G_1 \subseteq [m]$. Further, by our classification based on $\AOracle$, it is easy to see that $\Oracle(\vecp_t)_i > 0$ for any $i \in G_1$. This means from $\vecp_{t-1}$ to $\vecp_t$, only prices of goods in $G_1$ are increased. Let $S = \{i \mid (\vecp_t)_i > \vecp^*_i\}$, then we have $S \subseteq G_1$.

  Next, we apply a sequence of price changes to $\vecp_t$. First, for every $j \notin S$ we increase $(\vecp_t)_j$ to $\vecp^*_j$. Let $\vecp$ be the new price vector and consider the surplus $\Oracle(\vecp)$ resulting from exact demands. By the WGS property of the market, the surplus of any good in $S$ will not decrease, hence we still have $\Oracle(\vecp)_i > 0$ for every $i \in S$. The sum of all surpluses in an exchange market is always 0, so $\sum_{j \notin S}\Oracle(\vecp)_j < 0$.

  Now we decrease the price of every good $i \in S$ from $\vecp_i$ to $\vecp^*_i$. Then $\vecp$ becomes exactly $\vecp^*$. This process will not increase surplus of any good $j \not\in S$. Thus we still have $\sum_{j \notin S}\Oracle(\vecp^*)_j < 0$. This contradicts the assumption that $\vecp^*$ are prices of a market equilibrium.
\end{proof}
At this point, let us recall Claim~\ref{claim:xp} to establish the relation between the surplus of a good before and after a multiplicative price update step. It does not involve the approximate oracle.

\medskip

{\noindent \bf Claim~\ref{claim:xp}}
{\it For any price vector $\vecp$, $x > 1$ and $S \subseteq [m]$, $\Oracle(\UpdatePrice(\vecp, x, S))_i \leq x \cdot \Oracle(\vecp)_i$ for any $i \in S$.}

\medskip

Next, we establish a statement about the surpluses at the end of each round, which was not necessary for the version with exact oracles and precision. Intuitively, we increase the prices of goods in $G_1$ until the minimum surplus in $G_1$ reaches the maximum surplus in $G_2$ or 0. Note that $\mu$ is very small and can be thought of as 0. The main complication here is that we need to work with $\mu$-approximation demands in the algorithm and the resulting surpluses $s'$.

\begin{claim} \label{append:claim:minG1}
  At the end of each round in \AlgExchange, $\min\{\ts'_i \mid i \in G_1\} \le \max\{\{\ts'_i \mid i \in G_2\}\cup
  \{\mu\}\} + 6\mu$.
\end{claim}

\begin{proof}
  According to the binary search procedure, we know that if we increase prices in $G_1$ by a factor of $x$, then $\ts'$ satisfies the condition $\min\{\ts'_i \mid i \in G_1\} \geq \max\{\{\ts'_i \mid i \in G_2\}\cup \{\mu\}\}$. Furthermore, an increase by $x^+= \min\{ y \in \mathcal{P} \mid y > x\} < x+\frac1{\Delta}$ would result in a surplus vector that does not satisfy this condition. Let $\vecs^+ = \Oracle(\UpdatePrice(\vecp_{t-1}, x^+, G_1))$. By Assumption~\ref{asp2}, we have
  $$|\ts^+_i - \ts'_i| \leq |s^+_i - s'_i|+2\mu < 2^{D_2}\cdot (x^+-x)|\vecp_{t-1}|+2\mu \le \frac{2^{D_2}|\vecp_{t-1}|}{\Delta}+2\mu \leq \frac{2^{D_2+D_1+\log{m}}}{\Delta}+2\mu = 3\mu$$
  for every $i$, where the last inequality is derived by Claim~\ref{append:claim:priceUpper}. Thus
\begin{eqnarray*}
   \min\{\ts'_i \mid i \in G_1\}  <  \min\{\ts^+_i \mid i \in G_1\} + 3\mu 
    & < & \max\{\{\ts^+_i \mid i \in G_2\} \cup \{\mu\}\} + 3\mu \\
    & < & \max\{\{\ts'_i \mid i \in G_2\} \cup \{\mu\}\} + 6\mu
  \end{eqnarray*}
\end{proof}

We are now ready for the key lemma in the proof of the main result -- the multiplicative decrease of the potential function at the end of each round. Let $R_2$ be a sufficiently large constant to be explained in the end of the proof of the following lemma.

\begin{lemma}\label{append:lemma:potential}
  If $x < 1+\frac{1}{R_2m^3}$ at the end of round $t$ in \AlgExchange, 
  then $\Phi(\vecp_t) \leq \Phi(\vecp_{t-1}) \left(1-\Omega\left(\frac{1}{m^3}\right)\right)$.
\end{lemma}

\begin{proof}
  We use the following notation. Let $\vecs = \Oracle(\vecp_{t-1})$, $\tvecs = \AOracle(\vecp_{t-1},
  \mu)$ and $\vecs' = \Oracle(\vecp_t)$, $\tvecs' = \AOracle(\vecp_t,
  \mu)$.
	The intuition of the proof is similar to the version with exact precision. By the conditions used to define $G_1$ and $G_2$, we always have $\tilde{s}_{i_k} \geq \tilde{s}_{i_1} /e$ and $\tilde{s}_{i_k} - \tilde{s}_{i_{k+1}} > \tilde{s}_{i_k}/(m+1) \geq \tilde{s}_{i_1} / e(m+1)$. Hence, roughly speaking, every good in $G_1$ has reasonably large surplus, and there is a reasonably large gap between the surpluses in $G_1$ and $G_2$. Next, at the end of the current round, we decreased the minimum surplus of a good in $G_1$ to either $\min\{\tilde{s}'_i \mid i \in G_1 \} \approx \mu$ (Case (1) below) or $\min\{\tilde{s}'_i \mid i \in G_1 \} \approx \max\{\tilde{s}'_i \mid i \in G_2\}$ (Case (2) below). In both cases, the total value of $\Phi$ must decrease by a factor of $1 - \Omega(1/m^3)$.

  More formally, if the algorithm proceeds to round $t$, then $\|\tvecs\| > \varepsilon'^2$. By the definition of set $G_1$, we have $\ts_{i_1} \leq (1+\frac1m)\ts_{i_2} \leq \cdots \leq (1+\frac1m)^{k-1}\ts_{i_k} < e\cdot \ts_{i_k}$. Hence, $\ts^2_{i_k} > (\ts_{i_1}/e)^2 \geq \Phi(\vecp_{t-1})/(me^2) > (\varepsilon'/(\sqrt{m}e))^2$, so the surpluses of goods in $G_1$ are similar up to a factor of $e$ and bounded from below. Also we have $(\tilde{s}_{i_k} - \tilde{s}_{i_{k+1}})^2 > (\tilde{s}_{i_1} / e(m+1))^2 \geq \Phi(\vecp_{t-1})/(e^2(m+1)^2m)$.

Since we rely on an approximate demand oracle, the surpluses of goods in $G_1$ might not change in a monotone fashion when increasing their prices. Nevertheless, we can relate the surplus in the beginning and the end of a round as follows. For every $i \in G_1$, by Claim~\ref{claim:xp}, the surplus from exact demands satisfies $s'_i \leq x\cdot s_i$. Thus 
  $$\ts'_i \; \leq \; s'_i + \mu \; \leq \; xs_i + \mu \; \leq \; x(\ts_i + \mu) + \mu \; = \; x\ts_i + (1+x)\mu.$$
Since $x < 1+\frac1{R_2m^3}$, it holds that $(1+x)\mu < 3\mu \leq \ts_i / (R_2m^3)$. This means the increase within a round is bounded by $\ts'_i < (1+\frac{2}{R_2m^3})\ts_i$. 
Since we do not touch the price of any good $j \in G_2$, the WGS property implies for exact demands $s'_j \geq s_j$. Hence $\ts'_j \geq \ts_j - 2\mu$. 
  
Now, in order to bound the change of $\Phi(\vecp_t)$, we consider $\tvecs'$ according to $G_1$ and $G_2$. We distinguish two cases.

\paragraph{Case 1: $\max\{\ts'_i \mid i \in G_2\} < \mu$} 
Intuitively, in this case the algorithm has decreased the surplus of some good in $G_1$ to approximately 0 (recall that $\mu$ is sufficiently small). This decrease alone brings down the potential function $\Phi$ by a factor of $1 - \Omega(1/m)$. All other surpluses will cause an increase by a factor of at most $1+O(1/m^3)$.

More formally, Claim~\ref{append:claim:minG1} gives us $\mu < \min\{\ts'_i \mid i \in G_1\} < 7\mu$. Hence, the contribution of goods of $G_1$ to $\Phi(\vecp_t)$ can be upper bounded by
{\small
$$\sum_{j=1}^k\ts'^2_{i_j} < \sum_{j=1}^{k-1}\left(1+\frac2{R_2m^3}\right)^2\ts^2_{i_j}+ 49\mu^2.$$}
Furthermore, for every $i \in G_2$, if $-m^3\mu \leq \ts'_i \leq \mu$, we have $\ts'^2_i \leq m^6\mu^2$, and if $\ts'_i < -\mu$, by the WGS property of the market, we know $s_i \leq s'_i \leq \ts'_i + \mu < 0$. Thus, since $\ts'_j \geq \ts_j - 2\mu$,
{\small
$$ \sum_{\substack{j \in G_2\\ \ts'_j < -m^3\mu}} \ts'^2_j  
\leq \sum_{\substack{j \in G_2\\ \ts'_j < -m^3\mu}} (\ts_j-2\mu)^2 
\leq \sum_{\substack{j \in G_2\\ \ts'_j < -m^3\mu}} \left(1+\frac{2}{R_2m^3}\right)^2\ts^2_j. $$}
Hence, the contribution of goods of $G_2$ to $\Phi(\vecp_t)$ can be upper bounded by
{\small
\begin{eqnarray*}
 \sum_{j=k+1}^m \ts'^2_{i_j}& \leq & \sum_{j=k+1}^m\max\left\{\left(1+\frac{2}{R_2m^3}\right)^2\ts^2_{i_j}, m^6 \mu^2\right\} \; < \; \sum_{j=k+1}^m\left(1+\frac{2}{R_2m^3}\right)^2\ts^2_{i_j} + m^6\mu^2.
\end{eqnarray*}}
Combining the two parts 
{\small
\begin{eqnarray*}
	\Phi(\vecp_t) = \sum_{j=1}^m\ts'^2_{i_j} & < & \sum_{j\neq k}\left(1+\frac{2}{R_2m^3}\right)^2\ts^2_{i_j}+(m^6+49)\mu^2\\
    &=&\left(1+\frac{2}{R_2m^3}\right)^2(\Phi(\vecp_{t-1})-\ts^2_{i_k})+(m^6+49)\mu^2\\
    &<&\left(1+\frac{2}{R_2m^3}\right)^2\left(1-\frac{1}{e^2m}\right)\Phi(\vecp_{t-1})+\left(\frac{4}{R_1^2m^7} + \frac{196}{R_1^2m^{13}}\right)\varepsilon'^2\\
    &<&\left(1-\frac{1}{2e^2m}\right)\Phi(\vecp_{t-1})
\end{eqnarray*}}
where the last inequality holds for any $m \ge 2$ with sufficiently large constants $R_1, R_2$.

\paragraph{Case 2: $\max\{\ts'_i \mid i \in G_2\} \ge \mu$} 
Intuitively, in this case the gap between surpluses in $G_1$ and $G_2$ decreases to approximately 0. Below we show that the closing this gap yields a decrease of the potential function $\Phi$ by a factor of $1 - \Omega(1/m^3)$. All other surpluses will increase by a factor of at most $1+O(1/m^3)$. In combination, it turns out that $\Phi$ will decrease by a factor of $1 - \Omega(1/m^3)$.

More formally, in this case $\min\{\ts'_i \mid i \in G_1\} \geq \max\{\ts'_j \mid j \in G_2\}$. Let $s_{G_1} = \min\{\ts'_i \mid i \in G_1\}$ and $s_{G_2}=\max\{\ts'_j \mid j \in G_2\}$.
For every $i \in G_1$, let $\ts'_i = x'\ts_i - \delta_i$ where $x'=(1+\frac2{R_2m^3})$, and for every $j \in G_2$, let $\ts'_j = \ts_j - 2\mu + \delta_j$. Hence $\delta_i, \delta_j \geq 0$ for all $i, j$. Further, we have $|\sum_{i=1}^m \ts_i| \leq m\mu$ and $|\sum_{i=1}^m \ts'_i| \leq m\mu$ , hence
\begin{eqnarray*}
\sum_{i \in G_1}\delta_i & = & \sum_{i \in G_1} x'\ts_i - \sum_{i \in G_1}\ts'_i \\
& \geq & \sum_{i \in G_1} \ts_i + \sum_{j \in G_2}\ts'_j - m\mu = \sum_{i \in G_1} \ts_i + \sum_{j \in G_2} (\ts_j + \delta_j - 2\mu) - m\mu \\
& \geq & \sum_{j \in G_2}\delta_j -4m\mu
\end{eqnarray*}
and 
\vspace{-0.4cm}
$$\sum_{i\in G_1}\delta_i \geq \frac12(\ts_{i_k}-\ts_{i_{k+1}}-4m\mu) \geq \frac1{4}(\ts_{i_k}-\ts_{i_{k+1}}).$$
Now we have 
{\small 
\begin{eqnarray}
	\Phi(\vecp_t) & = & \sum_i\ts'^2_i = \sum_{i \in G_1}(x'\ts_i - \delta_i)^2 + \sum_{j \in G_2}(\ts_j - 2\mu + \delta_j)^2 \nonumber\\
    & = & \left(\sum_{i \in G_1}x'^2\ts^2_i + \sum_{j \in G_2}(\ts_j-2\mu)^2\right) + \left(\sum_{j \in G_2}\delta_j(\ts_j-2\mu+\delta_j)-\sum_{i \in G_1}\delta_i(x'\ts_i-\delta_i)\right) - \sum_{i \in G_1}x'\ts_i\delta_i \nonumber \\
    & & \quad + \sum_{j \in G_2}\delta_j(\ts_j-2\mu) \label{append:eq:line1}\\
    & < & \left(x'^2\Phi(\vecp_{t-1}) -4\mu\sum_{j\in G_2}\ts_j + 4m\mu^2\right) + \left(s_{G_2}\sum_{j \in G_2}\delta_j - s_{G_1}\sum_{i \in G_1}\delta_i\right) - \ts_{i_k}\sum_{i \in G_1}\delta_i \nonumber \\
    & & \quad + (\ts_{i_{k+1}}-2\mu)\sum_{j \in G_2}\delta_j \label{append:eq:line2}\\   %
    & < & x'^2\Phi(\vecp_{t-1}) -4\ts_mm\mu + 4m\mu^2 + 4s_{G_1}m\mu + 4\ts_{i_{k+1}}m\mu - (\ts_{i_k} - \ts_{i_{k+1}})\sum_{i \in G_1}\delta_i\label{append:eq:line3}\\
    %
    %
    %
    %
    & < & x'^2\Phi(\vecp_{t-1}) + 4m\mu^2 + 24m^2\mu - \frac{1}{4}(\ts_{i_k}-\ts_{i_{k+1}})^2 \label{append:eq:line5}\\
    & < & \left(1+\frac4{R_2m^3}+\frac4{R_2^2m^6} + \frac{16}{R_1^2m^{12}} + \frac{96}{R_1m^4} -\frac1{4e^2(m+1)^2m}\right)\Phi(\vecp_{t-1}) \label{append:eq:line6}\\
    & = & \left(1-\Omega\left(\frac1{m^3}\right)\right)\Phi(\vecp_{t-1}) \label{append:eq:line7}
\end{eqnarray}
}
Here \ref{append:eq:line1} can be derived by expanding the quadratic formula and appropriately reorganizing the terms. For the step from \eqref{append:eq:line1} to \eqref{append:eq:line2}, in the first bracket we overestimate the quadratic terms of $\ts$ into $x'^2\Phi(\vecp_{t-1})$ and $|G_2|$ by $m$. In the second bracket, we return to $\ts_i'$ and $\ts_j'$, which in turn are bounded correctly using $s_{G_1}$ for all $i \in G_1$ and $s_{G_2}$ for all $j \in G_2$. For the final two terms in~\eqref{append:eq:line1} and \eqref{append:eq:line2} we use the definition of $\ts_{i_k}$ and $\ts_{i_{k+1}}$ and the fact that $x' > 1$. 
For the step from \eqref{append:eq:line2} to \eqref{append:eq:line3}, for the first bracket of \eqref{append:eq:line2} we use $\ts_j \geq \ts_m$ for every $j \in G_2$. For the second bracket of \eqref{append:eq:line2} we note $\ts_{G_1} \ge \ts_{G_2}$ and the difference between the sums of $\delta$-terms is bounded by $4m\mu$ as noted above. By the same argument, we can transform the last two terms of \eqref{append:eq:line2} as shown. Note that we simply drop $-2\mu\sum_{j \in G_2} \delta_j < 0$. From \eqref{append:eq:line3} to \eqref{append:eq:line5} we use the fact that every surplus is bounded by $2m$ in its absolute value by Claim~\ref{append:claim:surplusBound}. For the last term we use the bound for $\sum_{i \in G_1}\delta_i$ as noted above. From \eqref{append:eq:line5} to \eqref{append:eq:line6}, we just replace $\mu$ by its definition and use the bound $\Phi(\vecp_{t-1}) > \varepsilon'^2$ and $(\tilde{s}_{i_k} - \tilde{s}_{i_{k+1}})^2 > \Phi(\vecp_{t-1})/(e^2(m+1)^2m)$ as shown above.

Finally, using sufficiently large constants $R_1, R_2$, the multiplicative term in \eqref{append:eq:line6} can be decreased to strictly less than 1 for every $m \ge 2$. The final expression in~\eqref{append:eq:line7} captures the asymptotics and proves the lemma.
\end{proof}

The previous lemma shows a decrease in the potential only for rounds in which the $x$ determined by binary search is rather small. Lemma~\ref{lemma:xmax} continues to hold here and bounds the number of rounds with a larger value of $x$. The following variant differs only in the constant $R_2$, and its proof is literally the same as for Lemma~\ref{lemma:xmax}.

\begin{lemma}
	\label{append:lemma:xmax}
	During a run of \AlgExchange, there can be only $O(m^4 D_1)$ many rounds that end with $x \geq 1+\frac1{R_2m^3}$.
\end{lemma}

Finally, we can assemble the properties to show the main result.

\medskip

{\noindent \bf Theorem~\ref{thm:exchangeWGS}}
{\it For any market that satisfies Assumptions~\ref{asp1} and~\ref{asp2}, and for any $\varepsilon > 0$, \AlgExchange
  returns the price vector of an $(1+\varepsilon)$-approximate market equilibrium in time polynomial in the input size and $\log(1/\varepsilon)$.}

\smallskip

\begin{proof}
  Let $x_t$ be the value of $x$ we find in round $t$ of \AlgExchange. First because at least one price will increase by a factor of $x_t$ in round $t$, by Claim~\ref{append:claim:priceUpper} we have $\prod_tx_t \leq 2^{mD_1}$.
  At the end of round $t$, if $x_t \geq 1+\frac1{R_2m^3}$, let $s=\max\{s_i \mid s_i \in \AOracle(\vecp_{t-1})\}$ 
  and $s'=\max\{s_i \mid s_i \in \AOracle(\vecp_{t})\}$, 
  then by Claim~\ref{append:claim:xp} we have $\Phi(\vecp_t) \leq ms'^2 \leq mx_t^2s^2 \leq mx_t^2\Phi(\vecp_{t-1})$.
  Moreover, by Lemma~\ref{append:lemma:xmax} there will be at most $O(m^4D_1)$ such rounds.
  Hence the total increase of $\Phi(\vecp_t)$ in these rounds will be no more than a factor of 
  $\prod_{x_t \geq 1+1/R_2m^3}{mx_t^2} \leq m^{O(m^4D_1)}2^{2mD_1} = m^{O(m^4D_1)}$.

  For all other rounds, we have $x < 1+\frac1{R_2m^3}$, and by Lemma~\ref{append:lemma:potential}, the potential function is decreased by a factor of $1/(1-\Omega(\frac1{m^3}))$. Therefore the total number of rounds before $\Phi(\vecp_t) \leq \varepsilon'^2$ will be at most 
{\small
  $$O\left(\log_{1/(1-\Omega(\frac1{m^3}))}\frac{m^{O(m^4D_1)}}{\varepsilon'^2}\right) = O\left(D_1m^7\log{m}+m^3\log{\frac1\varepsilon}\right)$$
  }
  In each round, the number of queries to the oracle is no more than $O(\log{\Delta}) = O(D_1+D_2+\log{m}+\log{\frac1{\varepsilon}})$. We conclude that the total number of queries during the algorithm is $O(\left(D_1m^7\log{m}+m^3\log{\frac1\varepsilon}\right)(D_1+D_2+\log{m}+\log{\frac1\varepsilon}))$, which is a polynomial in the input size and $\log(1/\varepsilon)$.
\end{proof}

\section{Exact Equilibrium for Exchange Markets with Spending Constraint Utilities}
 \label{app:rounding}
Here we show how to obtain an exact market equilibrium in exchange spending constraint markets. Using the \AlgSCExact framework we convert the approximate market equilibrium obtained in Section~\ref{subsec:precision} into an exact equilibrium. To achieve this, we rely on the full information of the spending constraint utilities. This step is an extension of the technique developed in~\cite{DuanM15} for the linear exchange markets. However there are several challenges due to the much more involved setting of spending constraint utilities, where the allocated partitions make the remaining budgets of agents and the values of goods dependent on too many parameters. In the following we present how to handle the extra complexity of the problem and this result resolves an open question of~\cite{DuanM15} of finding an exact polynomial time algorithm for exchange markets with spending constraint utilities.

Let $\vecp$ be the price vector of an $(1+\varepsilon)$-approximate equilibrium. We first construct a bipartite graph $EG'(\vecp) = (A \cup G, E)$ where the edges $E$ is a union of equality edges in $EG(\vecp)$, edges due to positive endowments, and edges due to allocated segments. The main idea here is to construct a set of components of agents and goods such that there is no interaction across the components. In addition, we want at least one good with price 1 in every such component. To achieve this latter condition, whenever there is a component $C$ of $EG'(\vecp)$ without a good with price 1, we raise the prices of goods in $C$ by a common factor $x>1$ until a new equality edge appears. By Assumption \ref{a:suff} a new equality edge will always emerge in during this procedure because prices of goods in $C$ are increasing, which makes goods outside $C$ more and more attractive to the agents in $C$. 

The next lemma shows that the updated price vector after the price increase still remains a $(1+\varepsilon)$-approximate equilibrium. 

\begin{lemma}
The price vector $\vecp$ at the end of while loop in \AlgSCExact is a $(1+\varepsilon)$-approximate equilibrium. 
\end{lemma}

\begin{proof} 
Note that we increase prices of goods in a component $C$ when each good has price greater than one. Corollary \ref{cor:connected} implies that the surplus of each good in $C$ is zero. Hence, the old allocation will still be feasible after the price change and the surpluses remain the same. Therefore the updated $\vecp$ will remain a $(1+\varepsilon)$-approximate equilibrium. 
\end{proof}

At this stage we can assume that each component of $EG'(\vecp)$ has a good with price 1. We then work on each component of $EG'(\vecp)$ separately. We assume for convenience that $EG'(\vecp)$ is a single component. 

\begin{algorithm*}[t]
\caption{\AlgSCExact \label{alg:AlgSCExact}}
\DontPrintSemicolon
\SetKwInOut{Input}{Input}\SetKwInOut{Output}{Output}
\SetKw{Param}{Parameters:} 
\Input{Exchange market with a set $A$ of agents and a set $G$ of goods; $w_{ij}, u_{ijk}, B_{ijk}$ are market parameters as defined in Section
\ref{sec:model}}
\Output{Prices $\vecp$ of an exact market equilibrium}
\BlankLine
$m \ot |G|;\ n \ot |A|;$ 
$L \ot $ total bit length of all input parameters; 
$\varepsilon \ot 1/m^{4m} 2^{4m^2L}$\;
$\vecp \ot $ $(1+\varepsilon)$-approximate equilibrium using \AlgSC with \AlgSCRounding\;
$\vecs \ot \Oracle(\vecp)$.  
If $\vecs = (0, 0, \ldots, 0)$ then \Return{$\vecp$}. \; 
$EG(\vecp) \ot$ (undirected) equality graph at prices $\vecp$ /*as defined at the beginning of Section \ref{sec:spending}*/\;
$F \ot \{(i,j,k)\ | \ (i,j,k) \mbox{ is an allocated segment}\}$\;
$EG'(\vecp) \ot EG(\vecp) \cup \{(i,j)\  |\ w_{ij} > 0\} \cup \{(i,j)\ |\ (i,j,k)\in F \mbox{ for some $k$}\}$.\;
\While{$EG'(\vecp)$ contains a connected component $C$ that does not has a good with price $1$}{
Find the smallest $x>1$ such that $EG(\vecp) \subset EG(\UpdatePrice(\vecp,x,C))$.\;
$\vecp \ot \UpdatePrice(\vecp, x, C)$.\;
Recompute $EG(\vecp)$ and $EG'(\vecp)$.\;
}
/* Wlog we assume $EG'(\vecp)$ consists of only one connected component. If there are more than one, then apply the procedure below individually to each component */ \;
Let $C_1, \ldots, C_K$ be the connected components of $EG(\vecp)$\;
\label{AlgSCExact:eqn}Set up the following system of linear equations in price variable
{
\begin{enumerate}
\item $p'_i = 1$ for a good $i$ whose price is $1$ 
\item For each component $C_l, 1 \le l \le K$
  \begin{enumerate}
  \item $|C_l| -1$ linearly independent equations of the form $u_{ijk}p'_{j'} = u_{ij'k'}p'_j$, where $(i, j, k)$ and $(i, j', k')$ are the current segments.
  \item $\sum_{j \in C_l\cap G} p'_j - \sum_j R_{lj}p'_j = 0$, where $R_{lj} = \sum_{i\in C_l\cap A} w_{ij}(1-\sum_{(i,j',k) \in F; j'\not\in C_l\cap G}B_{ij'k}) + \sum_{i\not\in C_l\cap A}w_{ij}\sum_{(i,j',k)\in F; j'\in C_l\cap G} B_{ij'k}$
  \end{enumerate}
\end{enumerate}
}


$\vecp' \ot$ the solution of above system of equations\;
\Return{$\vecp'$}
\end{algorithm*}

Next we set up a system of linear equations in price variables of the form $A\vecp = b$, and show that the matrix $A$ has full rank. Finally, we will show that by perturbing the vector $b$ slightly we can get an exact equilibrium.  Consider the components of $EG(\vecp)$, i.e., after removing edges due to endowment and allocated segments from $EG'(\vecp)$. Let $C_1, \dots, C_K$ be the set of components of $EG(\vecp)$. In each $C_l, 1 \le l \le K$, all goods are connected with each other through a set of equality edges. Whenever there are two current segments $(i,j,k)$ and $(i,j',k')$ of the same agent $i$, we have the following relation between the prices of goods $j$ and $j'$: 
\begin{eqnarray}\label{eqn:mbb}
u_{ijk} p_{j'} = u_{ij'k'}p_j. 
\end{eqnarray}
This implies that for each component $C_l$, $|C_l \cap G| - 1$ of these equations are linearly independent, and there is essentially one free price variable. Further, since there is no money flow across components with respect to the current allocations, we have the budget balance condition for each component:
\begin{eqnarray}
\begin{aligned}
\text{ Remaining worth of goods - remaining budgets of agents (after allocated segments) } \\
\text{ = sum of surpluses}
\end{aligned}\notag
\end{eqnarray}
For component $C_l$, the condition reads
\[\sum_{j \in C_l \cap G} (p_j - \sum_{(i,j,k)\in F} B_{ijk} \sum_{j'} w_{ij'}p_{j'}) - \sum_{i\in C_l\cap A} (\sum_{j'} w_{ij'}p_{j'} - \sum_{(i,j,k) \in F} B_{ijk} \sum_{j'} w_{ij'}p_{j'}) = \sum_{j\in C_l\cap G}\varepsilon_j,\]
where $F$ is the set of allocated segments. 
%
Rearranging the above equation, we get:
%
%
%
\begin{eqnarray}\label{eqn:mc}
\sum_{j \in C_l\cap G} p_j - \sum_{j\in G} p_j R_{lj} = \sum_{j\in C_l\cap G}\varepsilon_j, \hspace{7cm} \\
\mbox{ where }  R_{lj} = \sum_{i\in C_l\cap A} w_{ij}(1-\sum_{(i,j',k) \in F; j'\not\in C_l\cap G}B_{ij'k}) + \sum_{i\not\in C_l\cap A}w_{ij}\sum_{(i,j',k)\in F; j'\in
C_l\cap G} B_{ij'k}\notag
\end{eqnarray}
Each $R_{lj}$ is a rational number with denominator at most $2^{2L}$, where $L$ is the total bit length of all input parameters $w_{ij}, u_{ijk}$ and $B_{ijk}$.
\begin{lemma}
	\label{lem:rlj} 
	For every $1 \leq l \leq K$ and $j \in C_l \cap G$, $0 \le R_{lj} \le 1$. For every $j \in G$, $\sum_l R_{lj} = 1$. 
\end{lemma}
\begin{proof}
	For each good $j$, we have $\sum_i w_{ij} = 1$. Further both $\sum_{(i,j',k) \in F; j'\not\in C_l\cap G}B_{ij'k}$ and $\sum_{(i,j',k)\in F; j'\in C_l} B_{ij'k}$ take values in $[0, 1]$, hence the first claim of the lemma follows. For the second claim, $\sum_l R_{lj} = \sum_i w_{ij} = 1$.
\end{proof} 
Let $M$ be the coefficient matrix of the system of equations \eqref{eqn:mc}. Then 
\[
M_{lj} = \left\{\begin{tabular}{cc} $1 - R_{lj}$ & if $j\in C_l\cap G$ \\
$-R_{lj}$, & \mbox{ otherwise } \end{tabular}\right.
\]
From Lemma~\ref{lem:rlj}, each column $j$ of $M$ has at most one positive entry, namely $M_{lj}$ where $j \in C_l\cap G$, and each column of $M$ sums to zero. There are in total $L$ equations of type~\eqref{eqn:mc}, one for each component. 
%
%
%
Next we eliminate the equations of type~\eqref{eqn:mbb}. Then there will be only one price variable per component. We designate a representative good for each component, say good $l$ for $C_l$. 
Then each price $p_j$ in $C_l$ is a constant multiple of price $p_l$ of good $l$. Let $p_j = \alpha_jp_l$, where $\alpha_j$ is a rational number whose numerator and denominator are products of at most $m$ $u_{ijk}$'s. Now we can rewrite the budget balance equation~\eqref{eqn:mc} for $C_l$ in terms of $L$ price variables as follows: 
\[
p_l \sum_{j \in C_l\cap G} \alpha_j  - \sum_{l'}p_{l'}\sum_{j \in C_{l'}\cap G} \alpha_j R_{lj} = \sum_{j\in C_l\cap G}\varepsilon_j\notag
\]
Let $T_{l} = \sum_{j \in C_l\cap G} \alpha_j, S_{ll'} = \sum_{j \in C_{l'}\cap G} \alpha_j R_{lj}$, and $\varepsilon_l = \sum_{j\in C_l\cap G} \varepsilon_j$. The above equation becomes: 
\begin{equation}
\label{eqn:mconep}
p_l T_l  - \sum_{l'}p_{l'}S_{ll'} = \varepsilon_l
\end{equation}
Let $N$ be the coefficient matrix of the system of equations \eqref{eqn:mconep}. Then 
\[
N_{ll'} = \left\{\begin{tabular}{cc} $T_l - S_{ll}$ & if $l = l'$ \\
$-S_{ll'}$, & \mbox{ otherwise } \end{tabular}\right.
\]
Since both $T_l$ and $S_{ll'}$ are rational numbers with denominator at most $2^{mL}$ and $2^{2mL}$ respectively, each $N_{ll'}$ is a rational number with denominator at most $2^{3mL}$. 
\begin{lemma}
	\label{lem:N}
	$0 \le N_{ll}$ for every $l$; $N_{ll'} \le 0$ for every $l \neq l'$; and $\sum_{l} N_{ll'} = 0$ for every $l'$. 
\end{lemma}
\begin{proof}
	The proof essentially follows using Lemma~\ref{lem:rlj}. The first two claims are straightforward, and for the last claim we have 
\begin{eqnarray}
\begin{aligned}\notag
 \sum_l S_{ll'} = \sum_l \sum_{j \in C_{l'}\cap G} \alpha_j R_{lj}  = \sum_{j \in C_{l'}\cap G} \alpha_j \sum_l R_{lj} = \sum_{j \in C_{l'}\cap G} \alpha_j = T_{l'}
\end{aligned}
\end{eqnarray}
\end{proof} 
Since there is a good $i$ with price $1$, we assume with loss of generality that good $i$ belongs to component $K$, hence $\alpha_K p_K = 1$. The next lemma is an adaptation of a result of~\cite{DuanM15}. 
%
\begin{lemma}
The $K$ equations consisting of equation~\eqref{eqn:mconep} for components $1, \ldots, K-1$ and the equation $\alpha_K p_K = 1$ are linearly independent. 
\end{lemma}
\begin{proof}
Let $N'$ be the coefficient matrix of this system of equations. It is easy to check that $N'$ is the same as $N$ except that $N'_{Li} = 0,\ 1 \le i \le K-1$. Assume by contradiction that there is a non-zero vector $a = (a_1, \dots, a_K)$ such that $a^TN' = 0$. Let $a_{l_0}$ be the entry in $\{a_1, \dots, a_{K-1}\}$ that has the largest absolute value, and with loss of generality we assume that $a_{l_0} > 0$ and the first $K'$ entries of $a$ are equal to $a_{l_0}$, i.e., $a_{1} = \dots = a_{K'} = a_{l_0}$. 

For each $l \le K'$ we have
\begin{eqnarray}
\begin{aligned}\notag
 0 & = \sum_{1 \le h < K} a_h N_{hl}  + a_K \cdot 0 \\ 
   & = a_{l_0} \sum_{1 \le h \le K} N_{hl} - a_{l_0}N_{Kl} + \sum_{K' < h < K} (a_h - a_{l_0}) N_{hl}\\
   & = - a_{l_0}N_{Kl} + \sum_{K' < h < K} (a_h - a_{l_0}) N_{hl}\\
\end{aligned}
\end{eqnarray}
Using Lemma~\ref{lem:N}, the above implies that $N_{hl} = 0$ for $K' < h \le K$ and $l \le K'$. Next we show that $N_{lh} = 0$ for $1 \le l \le K'$ and $K'< h \le K$ as well. By summing up the equations~\eqref{eqn:mconep} for $1\le l\le K'$, we get
\begin{eqnarray}
\begin{aligned}\notag
 \sum_{l \le K'} \varepsilon_l 
   & = \sum_{l \le K'} N_{ll}p_l  + \sum_{h} \sum_{l \le K'; l \neq h} N_{lh}p_h\\ 
   & = \sum_{l \le K'} N_{ll}p_l  + \sum_{h \le K'} \sum_{l \le K'; l \neq h} N_{lh}p_h + \sum_{h > K'} \sum_{l \le K'} N_{lh}p_h\\ 
   & = \sum_{h \le K'} N_{hh}p_h  + \sum_{h \le K'} \sum_{l \neq h} N_{lh}p_h + \sum_{h > K'} \sum_{l \le K'} N_{lh}p_h\\ 
   & = \sum_{h \le K'} p_h \sum_l N_{lh}  + \sum_{h > K'} \sum_{l \le K'} N_{lh}p_h\\ 
   & = \sum_{h > K'} \sum_{l \le K'} N_{lh}p_h.\\ 
\end{aligned}
\end{eqnarray}
Since $p_h \ge 1$ for every $h$, if some $N_{lh}$ is non-zero for $l \le K'$ and $h > K'$, then the right-hand side is at most $-1/2^{3mL}$, which is a contradiction. It implies that $N_{lh} \neq 0$ iff both $l$ and $h$ are either less than or equal to $K'$ or larger than $K'$. Let $A_1$ and $G_1$ denote the set of agents and goods of components $C_1, \ldots, C_{K'}$ respectively. Let $A_2 = A \setminus A_1$ and $G_2 = G \setminus G_1$. We can further conclude that 
\begin{itemize}
\item $w_{ij} = 0$ for every $i \in A_1,\ j \in G_2$, otherwise $N_{lh}\neq 0$ for $l\le K'$ and $h> K'$. Similarly $w_{ij} = 0$ for every $i \in A_2,\ j \in G_1$. 
\item Agents in $A_1$ have no allocated goods in $G_2$, otherwise budget balance equations of components $C_l, l> K'$ will have some non-zero $p_j, 1
\le j \le K'$ and that will make $N_{lh} \neq 0$ for $l> K'$ and $h \le K'$. Similarly agents in $A_2$ have no allocated goods in $G_1$. 
\end{itemize}
This is impossible since we have assumed that $EG'(\vecp)$ consists of one single component. Therefore $N'$ must have full rank. 
\end{proof}
Overall, we have established that equations of~\eqref{eqn:mbb},~\eqref{eqn:mc} and $p_i = 1$ are linearly independent. We can write this system in the matrix form as $A\vecp = b$, where $A$ is invertible and all entries are rational numbers with common denominator at most $2^{2L}$. 

Consider the system $A\vecp' = b'$ for a price vector $\vecp'$, where $b'$ is a unit vector with a one in the row corresponding to the equation $p_i = 1$. Next we show that $\vecp'$ gives an exact equilibrium. For that we need to show the following:
\begin{itemize}
\item[(a)] Equality edges with respect to $\vecp'$ and $\vecp$ are the same. This will imply that all allocated segments remain allocated. 
\item[(b)] $(s, A\cup G\cup t)$ is a min-cut in $N(\vecp')$. Combining the equation(2b) in step~\ref{AlgSCExact:eqn} of \AlgSCExact, this will imply that there is a feasible allocation on current
segments which gives surplus of each good. 
\end{itemize}
By Cramer's rule and Lemma~\ref{lem:rlj}, the solution of $A\vecp' = b'$ is a vector of rational numbers with common denominator $D \le m^m2^{2m(m+1)L}$. That is, all $p_i'$ are of form $q_i/D$, where $q_i, D$ are integers. Since $||b| - |b'|| \le 2\varepsilon$, we have $|p_i'- p_i| \le 2\varepsilon D$ for every $i$. Let $\varepsilon' = 2D^2\varepsilon$, then $|D p_i - q_i| = D |p_i - p_i'| \le 2 \varepsilon D^2 = \varepsilon'$. For part (a), suppose $u_{ijk}p_{j'} \le u_{ij'k'}p_j$ then 
\begin{eqnarray}
\begin{aligned}\notag
u_{ijk}q_{j'}  \le u_{ijk} (Dp_{j'} + \varepsilon') \le D u_{ij'k'}p_j + u_{ijk}\varepsilon' & \le u_{ij'k'}q_j + (u_{ij'k'} + u_{ijk})\varepsilon'\\
& < u_{ij'k'}q_j + 1
\end{aligned}
\end{eqnarray}
Since both $u_{ijk}q_{j'}$ and $u_{ij'k'}q_j$ are integers, we have $u_{ijk}p'_{j'} \le u_{ij'k'}p'_j$. This implies that all equality edges with respect to $\vecp'$ and $\vecp$ are the same. That further implies that all allocated segments remain allocated. 

For the part (b), consider the network $N(\vecp')$ with respect to prices $\vecp'$. Recall that in $N(\vecp')$
\begin{itemize}
\item The capacity of edge from source node $s$ to agent $i$ is $\sum_{j'} w_{ij'}p'_{j'} (1 - \sum_{j,k: (i,j,k)\in F} B_{ijk})$. 
\item The capacity of a MBB edge from agent $i$ to good $j$ is $B_{ijk}\sum_{j'}w_{ij'}p'_{j'}$. 
\end{itemize}

Since $p'_{j}$'s are all rational numbers with a common denominator $D \le m^m2^{2m(m+1)L}$, all capacities in
$N(\vecp')$ are rational numbers with a common denominator no more than $D^2$. Also because $|p'_{i} - p_i| \le 2\varepsilon D$, the capacity of each edge $e$ in $N(\vecp)$ is at most the capacity of $e$ in $N(\vecp')$ plus $4\varepsilon D$. 

Let $c$ be the capacity of cut $(s, A\cup G\cup t)$ in $N(\vecp')$. Suppose there is a min cut in $N(\vecp')$ with value less than $c$. Then that value is at most $c-1/D^2$. This same cut in $N(\vecp)$ will have value at most $c-1/D^2 + (m+n+mn)4D\varepsilon$. Also the capacity of the cut $(s, A\cup G\cup t)$ in $N(\vecp)$ is at least $c - 4n\varepsilon D$. Therefore the total surplus of goods in $N(\vecp)$ is at least 
\[c - 4n\varepsilon D - \left(c-\frac{1}{D^2}\right) - (m+n+mn)4\varepsilon D = \frac{1}{D^2} - 4(m+2n+mn)\varepsilon D > \varepsilon,\]
which is a contradiction. Hence condition (b) also holds.

We conclude with our main theorem.

\vspace{2ex}

{\noindent \bf Theorem~\ref{thm:exchangelinear}}
{\it  For any spending constraint exchange market satisfying Assumption \ref{a:suff}, \AlgSCExact returns the price vector of a market equilibrium in time polynomial in $m$ and $L$.}

\begin{proof}
By Lemma~\ref{lem:apxlinear} we know that we arrive at a $(1+\varepsilon)$-approximate equilibrium in time polynomial in $m$, $L$, and $\log(1/\varepsilon)$.
For the exact equilibrium, \AlgSCExact solves a system of $K\le m$ linear equations whose entries are polynomially bounded in $m$ and $L$, hence can be done in polynomial time. Further, $\log(1/\varepsilon)$ is a polynomial of $m$ and $L$ as $\varepsilon = 1/m^{4m} 2^{4m^2L}$. Since $M$ is a polynomial of $m$, $L$, binary search and \Rounding run in polynomial time in each round of the framework. Hence, the total running time of the algorithm is polynomial in $m$ and $L$.
\end{proof}

\end{document}